\documentclass[iicol]{sn-jnl}   

\usepackage{graphicx}                        
\usepackage{amsmath,amssymb,amsfonts,amsthm} 
\usepackage{algorithm}                       
\usepackage{algpseudocode}                   
\usepackage{booktabs,multirow}               
\usepackage{xcolor}
\usepackage{bbm}   
\usepackage[round,authoryear]{natbib} 

\DeclareMathOperator{\Exp}{Exp}
\DeclareMathOperator{\Ber}{Ber}
\DeclareMathOperator{\Multinom}{Multinom}

\newcommand{\Copt}{C^{\scriptscriptstyle\textnormal{opt}}}
\newcommand{\Ctpp}{C^{\scriptscriptstyle\textnormal{tpp}}}
\newcommand{\Ctot}{C^{\scriptscriptstyle\textnormal{tot}}}
\newcommand{\indic}[1]{\mathbbm{1}_{#1}}
\newcommand{\lambdazz}{\lambda_{\text{ZZ}}}
\newcommand{\lambdabps}{\lambda_{\text{BPS}}}
\newcommand{\NextEvent}{\textsf{NextEvent}}
\newcommand{\UpperBound}{\textsf{UpperBound}}
\newcommand{\sprod}[2]{\langle #1,#2\rangle}
\newcommand{\tmax}{t_{\textnormal{max}}}
\newcommand{\tauopt}{\tau^{\textnormal{opt}}}
\newcommand{\Rset}{\mathbb{R}}
\newcommand{\dif}{\mathrm{d}}
\newcommand{\expd}{\Exp}   
\newcommand{\revision}[1]{#1}

\theoremstyle{plain}
\newtheorem{theorem}{Theorem}

\theoremstyle{definition}

\theoremstyle{remark}
\newtheorem{remark}{Remark}

\graphicspath{{Figures/}}

\begin{document}

\title[Automated Techniques for Efficient Sampling of Piecewise-Deterministic Markov Processes]{Automated Techniques for Efficient Sampling of Piecewise-Deterministic Markov Processes}

\author*[1]{\fnm{Charly} \sur{Andral}}\email{andral@ceremade.dauphine.fr}

\author[2]{\fnm{Kengo} \sur{Kamatani}}\email{kamatani@ism.ac.jp}

\affil*[1]{\orgdiv{CEREMADE}, \orgname{CNRS, Université Paris-Dauphine, Université PSL}, \orgaddress{\postcode{75016} \city{Paris}, \country{France}}}

\affil[2]{\orgname{Institute of Statistical Mathematics}, \orgaddress{\city{Tokyo}, \country{Japan}}}

\abstract{Piecewise deterministic Markov processes (PDMPs) are a class of continuous-time Markov processes that have recently been used to develop a new class of Markov chain Monte Carlo algorithms. However, implementing these processes is challenging because of their continuous-time nature and the need to integrate the rate function. Recently, \cite{corbellaAutomaticZigZagSampling2022} proposed a new algorithm that automates the implementation of the Zig-Zag sampler. However, the efficiency of the algorithm depends strongly on a hyperparameter ($\tmax$) that remains fixed throughout the run and must be tuned through preliminary runs. In this work, we relax this constraint and propose a new variant of the algorithm that allows this parameter to change over time and adapt automatically to the target distribution. We also replace Brent's optimization algorithm with a grid-based method for computing an upper bound on the rate function. This method is more robust to irregularities in the function and produces a tighter upper bound more quickly. We also extend the algorithm to other PDMPs and provide a Python implementation based on JAX.}


\maketitle

\section{Introduction}
\label{sec:intro}

Markov chain Monte Carlo (MCMC) has been a cornerstone of simulation-based inference for the last 30 years. However, most methods are based on the Metropolis--Hastings algorithm, which creates a reversible discrete-time Markov chain by adding a rejection step. This leads to random-walk behavior that can be inefficient in some settings. This observation led \cite{HOROWITZ1991247} to use partial velocity refreshment for Hybrid Monte Carlo. At the end of the 1990s, several authors \cite{diaconisAnalysisNonreversibleMarkov2000,chenLiftingMarkovChains1999} proposed new methods based on nonreversible Markov chains and showed that they could be more efficient than reversible methods.

The idea of nonreversible Markov chains was further developed by physicists in the late 2000s and early 2010s through different versions of Event-Chain Monte Carlo \cite{bernardEventchainMonteCarlo2009,michelGeneralizedEventchainMonte2014}, which were first applied to the hard-sphere model, and through the work of \cite{peters2012rejection}.

Finally, statisticians developed nonreversible Markov chains as an alternative to classical MCMC \cite{GUSTAFSON1998, bierkensNonreversibleMetropolisHastings2016,Kamatani2023}. \revision{Piecewise deterministic Markov processes (PDMPs), introduced by  \cite{davisPiecewiseDeterministicMarkovProcesses1984,davisMarkovModelsOptimisation1993}, have also been used to construct nonreversible MCMC algorithms.} \cite{bouchard-coteBouncyParticleSampler2018} extended the Bouncy Particle Sampler, which was originally proposed by \cite{peters2012rejection}. The discrete-time lifted Markov chain was extended to continuous time by \cite{bierkensPiecewiseDeterministicScaling2017} and then to multidimensional targets by the Zig-Zag Sampler \cite{bierkensZigZagProcessSuperefficient2019}. Both the Zig-Zag Sampler and the Bouncy Particle Sampler are continuous-time and rejection-free. Many other PDMP schemes were subsequently proposed, including Forward Event-Chain Monte Carlo \cite{michelForwardEventChainMonte2020}, the Boomerang Sampler \cite{bierkensBoomerangSampler2020}, the Coordinate Sampler \cite{wuCoordinateSamplerNonreversible2020}, and the Speed-Up Zig-Zag \cite{vasdekisSpeedZigZag2022}.

\revision{
A PDMP is a Markov process that evolves deterministically between event times. A major computational difficulty for PDMP samplers is the simulation of event times. Existing approaches can be divided broadly into thinning-based methods, which simulate the process exactly when a valid upper bound of the rate function is available, and numerical approximation methods. \cite{corbellaAutomaticZigZagSampling2022} automate the construction of a constant thinning bound, whereas \cite{suttonConcaveConvexPDMPbasedSampling2023} construct adaptive piecewise-linear bounds from a concave--convex decomposition of the rate. \cite{hardcastleAveragingPolyhazardModels2025} develop an application-specific adaptation of Automatic Zig-Zag, while \cite{zigzagboomerangjl}, accompanying \cite{bierkensBoomerangSampler2020}, implement analytic or model-specific local bounds. By contrast, \cite{paganiNuZZNumericalZigZag2024} use numerical integration and root finding, with error control in terms of the numerical tolerances, while \cite{bertazziApproximationsPiecewiseDeterministic2022} study general discretization schemes and their convergence as the step size tends to zero.
}

\revision{
The goal of this paper is to extend the automatic Zig-Zag algorithm from \cite{corbellaAutomaticZigZagSampling2022}. In contrast to the optimization-based strategy of \cite{corbellaAutomaticZigZagSampling2022}, which requires finding a maximal event rate via numerical optimization at each cycle, we propose a grid-based adaptive bound. Our method evaluates the rate and its derivative at a set of time grid points to form a piecewise-constant upper envelope. This avoids costly global optimizations while adapting the bound on the fly. Our construction does not require a concave--convex decomposition.
}

We apply this new method to PDMP samplers in Bayesian computation, where the event rate arises from a smooth log-density; however, the method \revision{can} be applied to any PDMP for which the flow is integrable and the event rate is \revision{sufficiently smooth}.
We provide a Python implementation of the algorithm based on JAX \cite{jax2018github}, making it easy to use and extend to other PDMPs. It is available on GitHub at \url{https://github.com/charlyandral/pdmp_jax} and on PyPI at \url{https://pypi.org/project/pdmp-jax}.

We derive sufficient theoretical conditions under which the grid-based adaptive bound is guaranteed to dominate the true rate (Theorem \ref{theorem:exactness}), so the process can be simulated exactly. However, in complex target distributions that violate these conditions (e.g., highly multimodal or nonsmooth targets), our approximate bound may underestimate the true rate in some intervals (Section \ref{sec:numerical_experiments}). In such cases, an event could be missed by the thinning sampler. To mitigate this, we build a safeguard into our algorithm. \revision{
When a violation is detected in adaptive mode, the sampler shortens the time horizon and
recomputes the bound; the scope and limitations of this runtime diagnostic
are discussed in Section~\ref{sec:bound_errors}.
} We found in practice that these issues are rare and have minimal impact on performance, even in challenging scenarios, thanks to the built-in safety checks. 

The outline of the paper is as follows. In Section \ref{sec:pdmp}, we present the general framework of PDMPs and some examples of PDMPs that will be used in the experiments. In Section \ref{sec:automatic_zigzag}, we present the automatic Zig-Zag sampler from \cite{corbellaAutomaticZigZagSampling2022} and analyze the role of the hyperparameter $\tmax$. In Section \ref{sec:general_method}, we present our new method to adapt $\tmax$ and to compute the upper bound of the rate function, while providing a theoretical proof of the correctness of the method. 
 In Section \ref{sec:numerical_experiments}, we present the results of the experiments we conducted to compare the efficiency of the new method with the automatic Zig-Zag sampler.

\section{Piecewise-deterministic Markov processes}
\label{sec:pdmp}

\subsection{Definition of a piecewise-deterministic Markov process}

Let $\Pi$ be a probability measure on $\mathbb{R}^d$ that admits a density $\pi(x)$ with respect to Lebesgue measure. We are interested in sampling from $\Pi$. Define the potential $U(x)$ of $\Pi$ such that $\pi(x) \propto \exp(-U(x))$. As in Markov chain Monte Carlo, the density need only be known up to a multiplicative constant, which means that the potential $U$ need only be known up to an additive constant. This does not affect the process because the potential is used only through its gradient.

 The location space $\mathbb{R}^d$ is extended by adding a velocity component that takes values in a velocity space $\mathcal{V} \subseteq \mathbb{R}^d$. This velocity space can be discrete (e.g., $\mathcal{V}= \{-1,1\}^d$) or continuous ($\mathcal{V} = \mathbb{R}^d$ or $\mathcal{V} = \mathbb{S}^{d-1}$).

 The goal is to create a continuous-time process $Z_t = (X_t,V_t)$ on $\mathbb{R}^d \times \mathcal{V}$ that admits $\Pi$ as the first marginal of its invariant distribution $\rho(\dif x\dif v)$. \revision{In the setting considered here}, a piecewise deterministic Markov process (PDMP) is a continuous-time Markov process defined by three elements:     

 \begin{enumerate}
    \item Deterministic dynamics that govern the process between two events. These dynamics are defined by a differential equation $\partial_t Z_t = \Psi(Z_t)$ for some function $\Psi$. This ODE corresponds to a flow function $\phi_t(x,v)$ that gives the state of the process at time $t$ given its location and velocity at time $0$. For instance, for the ZZS and the BPS, we have $\partial_t (X_t,V_t) = (V_t,0)$ and thus $\phi_t(x,v) = (x + t v,v)$.
    \item \revision{The event rate. Events occur with state-dependent intensity $\lambda(Z_{t-})$. In the samplers considered below, the position remains continuous and only the velocity is updated. Using the flow function, we write $\lambda(t;x,v)$ for $\lambda(\phi_t(x,v))$, or simply $\lambda(t)$ if the starting point is not relevant.}
    \item \revision{The transition kernel. At an event, the state jumps according to a transition kernel. For the velocity-jump samplers considered in this paper, we write this kernel as $Q(z,\dif v')$, where $v'$ is the new velocity.}
 \end{enumerate}

 \newcommand{\gen}{\mathcal{L}}
    \revision{For the velocity-jump setting considered in this paper,} the extended generator of the process at $z = (x,v)$ is given by 
    \begin{align*}
        \gen f(z) &= \nabla f(z) \cdot \Psi(z)\\
        &\quad+ \lambda(z) \int_{\mathcal{V}} \left( f(x,v') - f(x,v) \right) Q(z,\dif v')
    \end{align*}

    Defining \revision{such} a PDMP is therefore equivalent to giving the triplet $(\phi_t,\lambda,Q)$ (or $(\Psi, \lambda,Q)$).

\subsection{Some examples of PDMPs}
We briefly present several PDMP samplers considered in the experiments by giving the expressions for $\lambda$, $\Psi$, and $Q$ for each sampler.
\subsubsection{The Zig-Zag sampler}
We recall the expression of $\lambda$, $\Psi$ and $Q$ for the Zig-Zag sampler \cite{bierkensZigZagProcessSuperefficient2019}.  The velocity space is $\mathcal{V} = \{-1,1\}^d$, and its invariant measure is $\Pi \otimes \mathcal{U}(\mathcal{V})$. \revision{The deterministic dynamics are given by} $\Psi(z) = (v,0)$ which corresponds to the flow function $\phi_t(x,v) = (x + t v,v)$.
 The rate function is given by 
\newcommand{\underlambda}{\underline{\lambda}}
\begin{equation*}
    \lambda(x,v) = \sum_{i=1}^d (\nabla_i U(x) v_i)_+ =: \sum_{i=1}^d \lambda_i(x,v)
\end{equation*}
\revision{where $(a)_+ := \max\{a,0\}$ denotes the positive part of
$a\in\mathbb{R}$. }
 We denote by $F_m(v)$ the velocity obtained by flipping the $m$th component of $v$; for example, for $v = (1,-1,1)$, $F_2(v) = (1,1,1)$. The transition kernel of the ZZS is given by

 \begin{equation*}
        Q(x,v,\dif v') =  \sum_{i=1}^d \frac{\lambda_i(x,v)}{\lambda(x,v)} \delta_{F_i(v)}(\dif v')
\end{equation*}
where $Q(x,v,\dif v')=\delta_v(\dif v')$ if $\lambda(x,v) = 0$.

\subsubsection{The Bouncy Particle sampler}
For the Bouncy Particle Sampler \cite{peters2012rejection, bouchard-coteBouncyParticleSampler2018}, the velocity space is $\mathcal{V} = \mathbb{R}^d$, and its invariant measure is $\Pi \otimes \mathcal{N}(0,I_d)$. The deterministic dynamics are the same as for the ZZS; that is, $\phi_t(x,v) = (x + t v,v)$. The rate function is given by

\begin{equation*}
    \lambda(x,v) = \sprod{\nabla U(x)}{v}_+  + \underlambda
\end{equation*}
where $\underlambda$ is called the refreshment rate and is a constant that can be chosen to ensure the ergodicity of the process. 

To define the transition kernel, we need to define the reflection operator. For $v \in \mathbb{R}^d$, the reflection operator $R_v$ is defined by $R_v(w) = w - 2 \frac{\sprod{w}{v}}{\sprod{v}{v}}v$ and corresponds to the reflection of $w$ with respect to the hyperplane orthogonal to $v$. The transition kernel of the BPS is given by
\begin{align*}
    Q((x,v),\dif v') =& \frac{\sprod{\nabla U(x)}{v}_+}{\lambda(x,v)} \delta_{R_{\nabla U(x)}(v)}(\dif v')\\
    &+ \frac{\underlambda}{\lambda(x,v)} \mathcal{N}(\dif v';0,I_d). 
\end{align*}

\subsubsection{The Boomerang sampler}
The Boomerang sampler \cite{bierkensBoomerangSampler2020} is a variation of the Bouncy Particle Sampler. The velocity space is $\mathcal{V} = \mathbb{R}^d$, and its invariant measure is $\Pi \otimes \mathcal{N}(0,I_d)$. 
\revision{The deterministic dynamics are different from those of the BPS and are given by:}
\begin{equation*}
    \phi_t(x,v) = (x \cos(t) + v \sin(t),  -x \sin(t) + v \cos(t))
\end{equation*}

The rate function is the same as for the Bouncy Particle sampler but for $U'(x) = U(x) \revision{-} \frac{1}{2} \sprod{x}{x}$: $\lambda(x,v) = \sprod{\nabla U'(x)}{v}_+  + \underlambda,$ for a refreshment rate $\underlambda$. The transition kernel is also the same as for the BPS.
\subsubsection{Forward event-chain Monte Carlo}

Forward event-chain (FEC) Monte Carlo \cite{michelForwardEventChainMonte2020} is a variation of the Bouncy Particle Sampler in which velocity jumps are not limited to reflections. In the experiment, we consider the case in which the velocity space is the sphere $\mathcal{V} = \mathbb{S}^{d-1}$ with a uniform invariant distribution, so the invariant measure of the process is $\Pi \otimes \mathcal{U}(\mathcal{V})$. The dynamics are the same as for the BPS; that is, $\phi_t(x,v) = (x + t v,v)$.

The rate function is given by $\lambda(x,v) = \sprod{\nabla U(x)}{v}_+$, with no refreshment rate. The transition kernel is decomposed into two parts: one acts along the gradient of the potential, and the other acts on its orthogonal complement. \cite{michelForwardEventChainMonte2020} describe several possibilities for each part. For instance, one can choose to refresh the parallel component while slightly rotating the orthogonal component. This decomposition was also used in \cite{vanettiPiecewiseDeterministicMarkovChain2018, Wu2017}.

\subsection{Simulating PDMPs}

Because the flow is available in closed form in most cases, the main difficulty in simulating PDMPs is determining the time of the next event. For a process starting at $(x,v)$, this event time is governed by the time-varying rate $\lambda(t;x,v)$. It can be simulated by drawing a random variable $u \sim \expd(1)$ and finding the time $\tau$ such that
\begin{equation}
    \label{eq:time_of_next_event}
    \int_0^\tau \lambda(t;x,v)\dif t = u. 
\end{equation}
In practice and in the general case, this integral cannot be computed analytically. 

One way to circumvent this issue is to use the thinning algorithm. Suppose that we do not know how to integrate $\lambda$ but have an integrable upper bound $\Lambda$ such that $\lambda(t;x,v) \leq \Lambda(t)$ for all $t$. The thinning algorithm then provides a way to simulate the next event time. We simulate the first event $\tau^*$ from a Poisson process with rate $\Lambda$ and accept it with probability $\lambda(\tau^*)/\Lambda(\tau^*)$. If it is rejected, we simulate the next event from the Poisson process with rate $\Lambda$ and repeat the procedure.

\begin{theorem}{\cite{lewisSimulationNonhomogeneousPoisson1979}}
    \label{thm:thinning}
    Let a Poisson point process with inhomogeneous rate $\Lambda(t)$ be such that $\lambda(t) \leq \Lambda(t)$ for all $t$. The thinning algorithm described above generates a Poisson point process with rate $\lambda(t)$.
\end{theorem}

\section{The automatic Zig-Zag sampler from \cite{corbellaAutomaticZigZagSampling2022}}

\label{sec:automatic_zigzag}

\subsection{The algorithm}

The main contribution of \cite{corbellaAutomaticZigZagSampling2022} is to automate the thinning algorithm by automatically finding an upper bound $\Lambda$, which they choose to be constant, $\Lambda \in \mathbb{R}$. Because finding an upper bound for all $t$ is difficult, they propose computing it only over an interval $[0,\tmax]$ and using it locally. For a process at position $(x,v)$, this corresponds to using the upper bound

\begin{equation*}
    \Lambda  = \max_{t \in [0,\tmax]} \lambda(t;x,v).
\end{equation*}
If the time of the next event computed using $\Lambda$ is greater than $\tmax$, the process advances to time $\tmax$, and the upper bound is recomputed.
The hyperparameter $\tmax$ is fixed and must be tuned. As Figure 5 of their paper shows, its value strongly affects the number of gradient evaluations and thus the efficiency of the algorithm. The optimization over $t$ is performed using Brent's algorithm with a slight modification.

They assume that the function is monotonic on the interval $[0,\tmax]$, which is reasonable for the rate function when $\tmax$ is small enough that the process visits only one mode. They therefore propose checking whether the maximum is reached at $0$ or $\tmax$ by performing only one iteration of Brent's algorithm. If so, they stop; otherwise, they continue the optimization.

The full algorithm is presented in Algorithm \ref{alg:automatic_zigzag} for the special case of Zig-Zag.

\begin{algorithm}[htbp]
\footnotesize                                    
  \caption{Automatic Zig-Zag sampler \cite{corbellaAutomaticZigZagSampling2022}}
  \label{alg:automatic_zigzag}
  \begin{algorithmic}[1]
    \Require Initial state $(x_0,v_0)$; skeleton size $K$; rates $\lambda$; window $\tmax$
    \Ensure Skeleton $\{t_k,x_k,v_k\}_{k=1}^K$
    \State $t\gets0$, $k\gets1$, $(x,v)\gets(x_0,v_0)$
    \State $\bar\lambda\gets\displaystyle\max_{t\in[0,\tmax]}\lambda(t;x,v)$
    \While{$k\le K$}\label{alg:line:while:main}
      \State Draw $\tau^\ast\sim\Exp(\bar\lambda)$;\; $\tauopt\gets\tau^\ast$
      \State $u\gets0$
      \While{$\tauopt\le\tmax$ \textbf{and} $u=0$}\label{alg:line:while:inner}
        \State $\lambda(\tauopt)\gets\sum_{i=1}^d\lambda_i(\tauopt;x,v)$\label{alg:line:lambda:sum:computation}
        \State Draw $u\sim\Ber\!\bigl(\lambda(\tauopt)/\bar\lambda\bigr)$
        \If{$u=1$}\label{alg:line:if:accept}
          \State Draw $m\sim\Multinom(1:d,\{\lambda_i(\tauopt)/\lambda(\tauopt)\})$
          \State $t\gets t+\tauopt$;\; $x\gets x+\tauopt v$;\; $v\gets F_m(v)$
          \State Save $(t,x,v)$;\; $k\gets k+1$
          \State $\bar\lambda\gets\displaystyle\max_{t\in[0,\tmax]}\lambda(t;x,v)$
          \State Draw new $\tau^\ast$;\; $\tauopt\gets\tau^\ast$
        \Else\label{alg:line:if:reject}
          \State Draw $\tau^\ast\sim\Exp(\bar\lambda)$;\; $\tauopt\gets\tauopt+\tau^\ast$
        \EndIf
      \EndWhile
      \If{$\tauopt>\tmax$ \textbf{and} $u=0$}\label{alg:if:horizon:reached}
        \State \revision{$t\gets t+\tmax$;\; $x\gets x+\tmax v$;\; keep $v$ unchanged}
        \State Recompute $\bar\lambda$;\; draw new $\tau^\ast$;\; $\tauopt\gets\tau^\ast$
      \EndIf
    \EndWhile
  \end{algorithmic}
\end{algorithm}

\subsection{Analysis of the role of the hyperparameter $\tmax$}

The hyperparameter $\tmax$ is crucial for the efficiency of the algorithm, as noted before. In this section, we will interpret the influence of $\tmax$ in the two extreme cases where $\tmax$ is very small and very large. 

\subsubsection{Small $\tmax$ regime}
When $\tmax$ is very small, $\Lambda$ is optimized over a very short interval. The function therefore varies little, and the constant upper bound is relatively tight, leading to a high acceptance probability in the thinning algorithm. However, because the process cannot move for longer than $\tmax$, it frequently reaches the horizon without flipping (the condition on line \ref{alg:if:horizon:reached} of Algorithm \ref{alg:automatic_zigzag} is often true). The process advances by $\tmax$, and a new upper bound must then be computed. This leads to many gradient evaluations in the optimization step while the process moves very little.

\subsubsection{Large $\tmax$ regime}
When $\tmax$ is very large, the upper bound is computed over a large interval and becomes very loose. This leads to a very low acceptance probability in the thinning algorithm. The process will often need to propose many candidate events before accepting one.  Each time, the computation of the rate for the accept/reject step (line \ref{alg:line:lambda:sum:computation} of the algorithm) requires gradient evaluations. As the upper bound is kept the same unless an event is accepted or the horizon is reached, in this case, the loss of efficiency is due to the low acceptance probability.

\subsubsection{Counting the number of gradient evaluations}

As detailed in Equation 12 of \cite{corbellaAutomaticZigZagSampling2022}, the number of gradient evaluations can be decomposed into two parts. The first is the number of gradient evaluations used to optimize $\Lambda$, denoted by $\Copt$, and the second is the number used to compute the rate in the accept--reject step of the thinning algorithm, denoted by $\Ctpp$.

The total number of gradient evaluations $\Ctot$ is given by

\begin{equation*}
    \Ctot = \Copt + \Ctpp.
\end{equation*}

\revision{
The two regimes presented above correspond to this trade-off. When
$\tmax$ is small, the horizon is reached frequently and the upper bound must be recomputed by numerical optimization after each horizon hit. Consequently, the number of optimization calls, and hence $\Copt$, is large, whereas $\Ctpp$ is small. Conversely, when $\tmax$ is large, the upper bound is recomputed less frequently, but its looseness leads to more rejected candidate events and therefore to a larger $\Ctpp$. Choosing an appropriate value of $\tmax$ is therefore important for balancing these two sources of
computational cost.
}

\revision{This is illustrated in Figure \ref{fig:tmax:effect:non_adaptive}, where $\Copt$ and $\Ctpp$ are plotted as functions of the fixed value of $\tmax$ for the Zig-Zag sampler targeting a 30-dimensional standard Gaussian distribution.}

 \begin{figure}[h] 
     \centering
     \includegraphics[width=0.4\textwidth]{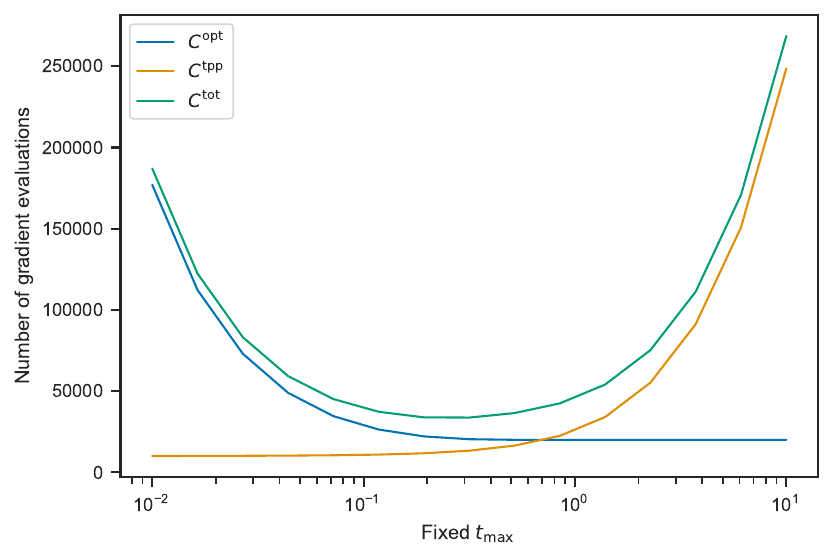}
     \caption{\revision{Fixed-horizon experiment: number of gradient evaluations as a function of the fixed value of $\tmax$. The $\tmax$ axis is logarithmic, while the vertical axis is linear.}}
     \label{fig:tmax:effect:non_adaptive} 
 \end{figure}
 
\subsection{The optimization problem}

The upper bound $\Lambda$ is computed using Brent's algorithm. However, when the target is multimodal, the resulting function $\lambda(t)$ may be nonconcave. In such cases, the optimization step may produce an invalid upper bound, biasing the simulated event times. The procedure is correct only as long as $\Lambda$ is a true upper bound on $\lambda$; otherwise, the simulated process does not have rate $\lambda$. This issue was not addressed by \cite{corbellaAutomaticZigZagSampling2022}, who assumed that the function is concave over a short interval (typically of length $\tmax$). However, one may wish to simulate PDMPs in more general cases, for example, when the target has several closely spaced modes and the function is nonconcave. Replacing Brent's algorithm with another optimization algorithm, such as one that uses the gradient of $\lambda$, may not solve the issue because the result may become trapped at a local maximum, depending on the starting point.

\subsection{\revision{Why is $\tmax$ constant?}}
\label{sec:adaptive_tmax}
The hyperparameter $\tmax$, as it appears in the algorithm, is purely numerical. This was also noted by \cite{suttonConcaveConvexPDMPbasedSampling2023}. It plays no role in the dynamics of the process, unlike, for instance, the standard deviation in a random-walk Metropolis--Hastings algorithm. It is used only to compute the upper bound $\Lambda$. Changing it each time the upper bound is recomputed does not affect the process, provided that the optimization is successful.

Following this observation, we propose to let $\tmax$ change over the execution of the algorithm. To do so, we need to understand where the algorithm gets stuck in the two extreme cases we have presented.

First, if $\tmax$ is very small, the process often reaches the horizon without an event. This means that the condition on line \ref{alg:if:horizon:reached} of Algorithm \ref{alg:automatic_zigzag} is often true. Whenever the horizon is reached without an accepted event, including after rejected candidates, we propose increasing $\tmax$ by multiplying it by a constant $\alpha > 1$.

In the second case, $\tmax$ is very large, and the entire computational budget is spent in the thinning algorithm due to the low acceptance probability. The algorithm will rarely satisfy the condition of the if statement at line \ref{alg:line:if:accept} and will mainly enter the else statement at line \ref{alg:line:if:reject}. Thus, in the else statement at line \ref{alg:line:if:reject}, we propose to decrease $\tmax$ by dividing it by $\alpha$.

\revision{
In both cases, $\tmax$ is updated geometrically. Here, an \textit{optimal} value of $\tmax$ refers to a value that approximately minimizes the computational cost per unit of simulated PDMP time by balancing the two sources of overhead described above. In our experiments, we assess this cost primarily through the number of gradient evaluations. We do not claim that this optimal value is unique, or that the adaptation formally converges to it; rather, the adaptation aims to keep $\tmax$ within a computationally
efficient range. A multiplicative update is used because it adjusts $\tmax$ on a relative scale, allowing it to move rapidly across orders of magnitude when its initial value is poorly chosen. Other update rules, such as diminishing adaptation, are also possible.
}

\revision{To illustrate the benefits of adaptation, we repeat the experiment from Figure \ref{fig:tmax:effect:non_adaptive} while allowing $\tmax$ to change over time, again targeting a 30-dimensional standard Gaussian distribution with a Zig-Zag sampler. We choose $\alpha = 1.1$. In Figure \ref{fig:tmax:effect:adaptive}, the horizontal axis gives the initial value of $\tmax$, which is subsequently adapted during the run. The number of gradient evaluations is almost constant across these initial values. Note that the range of initial values is also much larger than in the nonadaptive case.}

\revision{
An alternative horizon-adaptation strategy is proposed by \cite{suttonConcaveConvexPDMPbasedSampling2023}, who periodically
set $\tmax$ to a chosen quantile of previously simulated inter-event times.
By contrast, our multiplicative rule responds directly to the two sources of
computational inefficiency identified above: $\tmax$ is increased when the
horizon is reached without an accepted event and decreased when a proposed event is
rejected.
}

\begin{figure}[h]
     \centering
     \includegraphics[width=0.4\textwidth]{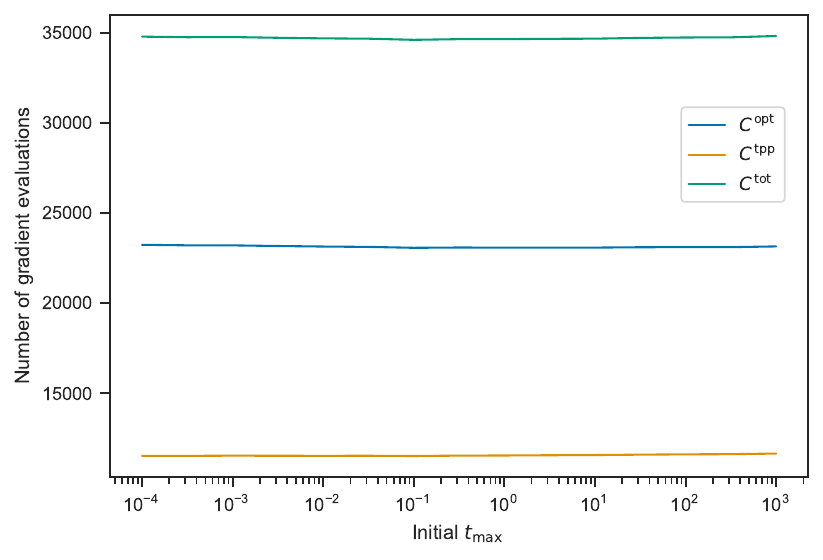}
     \caption{\revision{Adaptive-horizon experiment: number of gradient evaluations as a function of the initial value of $\tmax$, which is subsequently adapted during the run. The x-axis is the initial-$\tmax$ axis and is logarithmic, while the vertical axis is linear.} \label{fig:tmax:effect:adaptive}}
 \end{figure}

\section{New general method: using a grid}
\label{sec:general_method}
\subsection{New computation of the upper bound}
The thinning theorem is valid for a constant upper bound, which makes simulating the Poisson point process straightforward, but a constant bound is not necessary. We propose replacing the constant upper bound with a piecewise-constant one. Equation \eqref{eq:time_of_next_event} remains easy to solve when the function is piecewise constant, typically over a grid. This method has two main advantages.

First, this method produces a tighter upper bound than a constant bound. Using it in the thinning algorithm therefore leads to a much higher acceptance probability.

Second, this computation is much more robust to irregularities of the function. If the grid is sufficiently fine relative to the regularity of the function, the upper bound will be correct even when the function is nonconcave, resolving an issue with the original algorithm.

By leveraging the capabilities of the Python library JAX \cite{jax2018github}, particularly batched evaluation, we can evaluate the rate function $\lambda$ efficiently over a grid, unlike an optimization algorithm that evaluates the function one point at a time. Thus, even if the total number of evaluations is greater than in a conventional optimization step, the computation time may be shorter because the grid evaluations are performed in parallel. Moreover, unlike a conventional optimization algorithm, in which only the final evaluation is used, our method uses all evaluations of the rate function to construct the bound. We also leverage automatic differentiation in JAX to compute both the gradient of the potential and the derivative of the rate function.
\subsection{Construction of the upper bound}

\label{sec:general_definition_upper_bound}
We now present the construction of a piecewise-constant upper bound for a given differentiable function $\lambda$. As in the automatic Zig-Zag algorithm, we compute an upper bound over an interval $[0,\tmax]$ for a given point $(x,v)$. We introduce a new hyperparameter $N$, the number of segments in the grid. By default, we divide the interval $[0,\tmax]$ into $N$ equal segments $[t_i, t_{i+1}] :=[\frac{i}{N} \tmax, \frac{i+1}{N}\tmax]$ for $i=0,\dots,N-1$. We then construct an upper bound on $\lambda$ over each segment as follows.

On a given segment $[t_i, t_{i+1}]$ we have at our disposal four pieces of information: $\lambda(t_i;x,v)$, $\lambda(t_{i+1};x,v)$ and their time derivatives $\lambda'(t_i;x,v)$ and $\lambda'(t_{i+1};x,v)$, which we will denote, respectively, $y_i$, $y_{i+1}$, $d_i$ and $d_{i+1}$ for simplicity. The upper bound on the segment, denoted $\Lambda_i$, is taken to be the maximum of the three values:  $y_i$, $y_{i+1}$, and the ordinate of the intersection of the two tangents at $t_i$ and $t_{i+1}$, called $m_i$. 
The first two values are most relevant where the function is monotonic and the last one $m_i$, depending on the derivatives, is used to capture the case where the function reaches a local maximum on the segment. This is illustrated in
Figure \ref{fig:upper_bound_construction} where the function is the sine function, neither concave nor convex.

More precisely, the intersection point $(x_i,m_i)$ is the solution to the following system of equations:
\begin{equation}
    \begin{cases}
        m_i = d_i x_i + y_i - d_i t_i \\
        m_i = d_{i+1} x_i + y_{i+1} - d_{i+1} t_{i+1}
    \end{cases}
\end{equation}

This gives $x_i = \frac{y_{i+1} - y_i + d_i t_i - d_{i+1} t_{i+1}}{d_i - d_{i+1}}$ and $m_i = d_i x_i + y_i - d_i t_i$. A simplification is made: if $d_i = d_{i+1}$, $x_i = t_i$ and $m_i = y_i$. If $x_i$ does not lie in the interval, we clip it to the bounds.
The upper bound on the segment is then $\Lambda_i = \max(y_i, y_{i+1}, m_i)$.

Finally, on the whole interval $[0,\tmax]$, the upper bound $\Lambda$ can be expressed as $\Lambda(t) = \sum_{i=0}^{N-1} \Lambda_i \indic{[t_i, t_{i+1})}(t)$. 
This procedure is summarized in Algorithm \ref{alg:upper_bound}. Given a function $\lambda$, a horizon $\tmax$, and a number of segments $N$, we denote by $\UpperBound(\lambda,\tmax,N)$ the function that returns the upper bound $\Lambda$ using Algorithm \ref{alg:upper_bound} and the grid $(t_i)_{i=0}^{N} = (\frac{i \tmax }{N} )_{i=0}^{N}$.

\begin{algorithm}[htbp]
\footnotesize                                    
\caption{$\UpperBound$: construction of the upper bound}
\label{alg:upper_bound}
\begin{algorithmic}[1]                
  \Require Function $\lambda : \Rset \!\to\! \Rset$ to upper-bound; horizon $\tmax$; number of segments $N$
  \Ensure Upper bound $\Lambda$ of $\lambda$

  \For{$i \gets 0$ \textbf{to} $N$}
    \State $t_i \gets \dfrac{i}{N}\,\tmax$
    \State $y_i \gets \lambda(t_i)$ \Comment{evaluate $\lambda$ on the grid}
    \State $d_i \gets \lambda'(t_i)$  \Comment{and its derivative}
  \EndFor

  \For{$i \gets 0$ \textbf{to} $N-1$}
    \If{$d_i = d_{i+1}$}
      \State $x_i \gets t_i$
      \State $m_i \gets y_i$
    \Else
      \State $x_i \gets \dfrac{y_{i+1}-y_i + d_i t_i - d_{i+1} t_{i+1}}{\,d_i - d_{i+1}\,}$ \Comment{intersection abscissa}
      \State $x_i \gets \min\!\bigl(t_{i+1},\,\max(t_i,\,x_i)\bigr)$ \Comment{clip to bounds}
      \State $m_i \gets d_i\,x_i + y_i - d_i t_i$
    \EndIf
    \State $\Lambda_i \gets \max\!\bigl(y_i,\,y_{i+1},\,m_i\bigr)$
  \EndFor

  \State $\displaystyle \Lambda(t) \gets \sum_{i=0}^{N-1} \Lambda_i\,\indic{[t_i,\,t_{i+1})}(t)$
  \State \Return $\UpperBound(\lambda,\,\tmax,\,N) \;\leftarrow\; \Lambda$
\end{algorithmic}
\end{algorithm}

 \begin{figure}[htbp]
     \centering
     \includegraphics[width=0.4\textwidth]{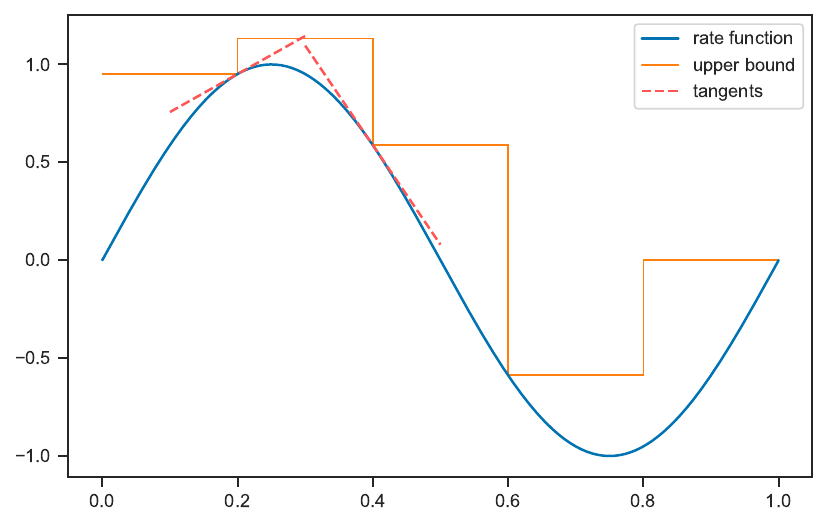}
     \caption{Construction of the upper bound for the sine function \label{fig:upper_bound_construction}}
 \end{figure}

\subsection{Correctness of the upper bound}

We will now show under which conditions the upper bound $\Lambda$ constructed in Algorithm \ref{alg:upper_bound} is indeed an upper bound of $\lambda$. Define 
\begin{align*}
    Z &=\{t \in (0,\tmax) : \lambda'(t) = 0 \text { and } \lambda''(t) < 0\}\\
    &\quad\cup\{t \in (0,\tmax) : \lambda''(t)=0\},
\end{align*}
the set containing the critical points of $\lambda$  that are (local) maxima, henceforth referred to as maximum critical points (they may also zeros of the second derivative),  and the zeros of its second derivative that are not maximum critical points. We also define the minimum distance between two points of $Z$ as 
$$\delta = \min_{t,t' \in Z, t \neq t'} |t-t'|.$$

\begin{remark}
    If the set $Z$ is infinite, typically because the first or second derivative vanishes over an interval, we can replace $Z$ by its quotient under the equivalence relation $t \sim t'$ when $t$ and $t'$ belong to the same connected component of $Z$. The condition $t \neq t'$ in the definition of $\delta$ then means that $t$ and $t'$ are not connected.
\end{remark}

\begin{remark}
    The intuition behind the definition of $Z$ is that, to obtain an upper bound, a local maximum and a change in concavity cannot occur within the same subinterval. This ensures that the function is concave near a local maximum and therefore lies below its tangents. Otherwise, one can easily construct a smooth function for which the function constructed in Algorithm \ref{alg:upper_bound} is not an upper bound because of the inflection point.
\end{remark}

Now that the set $Z$ and the minimum distance between its points $\delta$ are defined, we can state the main theorem for the correctness of the upper bound.

\begin{theorem}\label{theorem:exactness}\label{thm:upper-bound}
    Let $\lambda : [0,\tmax] \to \Rset$ be a twice continuously
     differentiable function such that $\delta>0$.
    Suppose that $\max_{i=0,\dots,N-1} |t_{i+1}- t_i| < \delta$. 
    
    Then, the function $\Lambda$ constructed in Algorithm \ref{alg:upper_bound} is an upper bound of $\lambda$, i.e., for all $t \in [0,\tmax]$, $\Lambda(t) \geq \lambda(t)$. 
\end{theorem}

In essence, the  conditions in Theorem \ref{theorem:exactness} require that, on every subinterval, the rate is either monotone or convex. This local regularity ensures that the piecewise-constant envelope serves as a valid upper bound for the rate on each subinterval.

\begin{proof}
    Let $t \in [0,\tmax]$. We will show that $\Lambda(t) \geq \lambda(t)$. Let $i$ be such that $t \in [t_i, t_{i+1})$. By construction of $Z$ and the grid, we have that $(t_i, t_{i+1})$ contains at most one point of $Z$. Thus we have two cases to consider: the interval $(t_i, t_{i+1})$ contains no maximum critical point of $\lambda$, or it contains one.

    If $(t_i, t_{i+1})$ contains no maximum critical point of $\lambda$, then, because $\lambda$ is differentiable on $(t_i, t_{i+1}]$, it reaches its maximum at the endpoints of the interval. Therefore, $\lambda(t) \leq \max(y_i, y_{i+1}) \leq \Lambda(t)$.

    In the second case, $(t_i, t_{i+1})$ contains a unique maximum critical point $s$ of $\lambda$. By definition of $Z$, $\lambda'' \neq 0$ on $(t_i, t_{i+1}) \backslash \{s\}$ and as $\lambda ''$ is continuous, this implies that $\lambda''$ does not change sign on $(t_i, t_{i+1})$ and thus that $\lambda'$ is monotonic on $(t_i,s]$ and $[s,t_{i+1})$. As $s$ is a local maximum, $\lambda'$ is decreasing on a neighborhood of $s$. Combining this with the previous argument, we have that $\lambda'$ is decreasing on $(t_i, t_{i+1})$ and thus $\lambda$ is concave on $(t_i, t_{i+1})$. The graph of $\lambda$ lies below its tangents. As the two tangents at $t_i$ and $t_{i+1}$ have slopes of opposite signs, the ordinate of the intersection of the two tangents is larger than $\lambda(s)$, the local maximum. Thus, $\Lambda(t) \geq \lambda(t)$.
  
\end{proof}

\begin{remark}
    In the specific case of a PDMP with a convex, twice-differentiable potential, the rate function is nondecreasing, and Algorithm \ref{alg:upper_bound} therefore always produces a valid upper bound, even for $N=1$. Indeed, if the rate function has the form $\lambda(t) = \sprod{\nabla U(x + tv )}{v}_+ + \underlambda$ (as in the Bouncy Particle or Forward Event Chain sampler), then the derivative of $\sprod{\nabla U(x + tv )}{v}$ is $\sprod{\nabla^2 U(x + tv )v}{v}$, which is positive because the Hessian is positive semidefinite.

    If the rate function is of the form $\lambda(t) = \sum_{i=1}^d (\nabla_i U(x + tv ) v_i)_+$ (for the Zig-Zag sampler), then the derivative of $\nabla_i U(x + tv ) v_i$ is $\sprod{\nabla^2 U(x + tv )v}{e_i}v_i = \sum_j \frac{\partial^2 U}{\partial x_i \partial x_j}(x + tv)v_j v_i$. Thus, for a Zig-Zag velocity, if the Hessian is diagonally dominant, then $\sum_j \frac{\partial^2 U}{\partial x_i \partial x_j}(x + tv)v_j v_i > 0$ and the rate function is increasing. This condition on the Hessian is not just a technical condition, and it was discussed in the Gaussian case in \cite{10.1214/18-AAP1453}.
\end{remark}

\begin{remark}
The conditions in Theorem \ref{theorem:exactness} are sufficient, though not necessary, for $\Lambda(t)$ to remain a valid envelope of $\lambda(t)$. If the interval is not short enough, we cannot guarantee $\Lambda(t)\ge\lambda(t)$ for all $t$, so a thinning proposal may arrive too late. At each proposed event time $T$, we compute $r=\lambda(T)/\Lambda(T)$. A valid envelope implies $r\le1$; if $r>1$, we flag a violation and, in adaptive mode, divide $\tmax$ by two and recompute the bound over the resulting shorter horizon, as described in Section \ref{sec:bound_errors}. This pointwise safeguard does not detect violations between evaluated times or retrospectively correct the simulated path. We illustrate in Appendix \ref{app:sensitivity:grid:tmax} how often detected violations occur and how they affect the estimation.
\end{remark}

\revision{In the remainder of the paper, whenever the dependence on $\lambda$ needs to be made explicit, we denote the function $\Lambda$ constructed in Algorithm~\ref{alg:upper_bound} by $B(\lambda)$; that is,
$$
\Lambda(t)=B(\lambda)(t).
$$
}

\subsection{Use of the upper bound in the case of PDMP sampling}

For most PDMP samplers, the rate function is either the Zig-Zag rate $\lambdazz(x,v) = \sum_{i=1}^d (\nabla_i U(x) v_i)_+$ or the Bouncy Particle rate $\lambdabps(x,v) = \sprod{\nabla U(x)}{v}_+ + \underlambda$. A simple solution would be to apply Algorithm \ref{alg:upper_bound} directly to $\lambdazz$ or $\lambdabps$. However, the positive-part operator makes the rate function nondifferentiable even when the potential $U$ is smooth, leading to errors in the bound. We therefore propose two solutions, depending on the rate function.

\subsubsection{Case of a Bouncy Particle-like rate}
\label{sec:bps_rate}

\revision{
Suppose that the process has a Bouncy Particle-like rate of the form
\[
\lambda(t)
=
(l(t))_+ + \bar{\lambda}.
\]
Here, $l(t)=\langle \nabla U(x_t),v_t\rangle$ for the BPS and
$l(t)=\langle \nabla U'(x_t),v_t\rangle$ for the Boomerang sampler.
We call the direct application of Algorithm~\ref{alg:upper_bound} to
$\lambda(t)$ the \emph{unsigned version}, which produces an upper bound
\[
\Lambda_{\mathrm{F}}(t)
=
B(l_+)(t)+\bar{\lambda},
\]
which takes the positive part of the function $l(t)$, and then constructs an upper bound on the resulting function.
The \emph{signed version} instead applies Algorithm~\ref{alg:upper_bound} to the raw
function $l(t)$, before taking its positive part. More precisely, 
\[
\Lambda_{\mathrm{T}}(t)=(B(l)(t))_+ + \bar{\lambda}
\]
is an upper bound for $\lambda(t)$.
Here, the subscripts $\mathrm{F}$ and $\mathrm{T}$ indicate that the signed version is disabled and enabled, respectively.
} \revision{
The main benefit of the signed version is that it retains gradient information even when $l(t)$ is negative, rather than producing a zero gradient because of the positive-part operator.} For instance, if the true rate function is concave and positive over an interval shorter than the grid spacing, computing an upper bound directly on $\lambda(t)$ fails, whereas computing an upper bound on $l(t)$ succeeds. \revision{In particular, $\Lambda_{\mathrm{F}}(t)\le \Lambda_{\mathrm{T}}(t)$. }

\subsubsection{Case of a Zig-Zag-like rate}
\label{sec:zigzag_rate}

For a rate function of the form $\lambdazz(t) = \sum_{j=1}^d (\nabla_j U(x_t) v_{t,j})_+$, the problem is more complex. In the previous case, the rate had only one point of nondifferentiability, where the rate function equals zero. By contrast, the Zig-Zag rate function can have points of nondifferentiability at any positive value. Computing an upper bound on $\sprod{\nabla U(x_t)}{v_t}$ does not work in this case because some terms of $\sum_{j=1}^{d} \nabla_j U(x_t) v_{t,j}$ may be negative and cancel one another.

\revision{
To distinguish the different strategies precisely, define
\[
l^{(j)}(t)=\nabla_j U(x_t)v_{t,j},
\]
so that
\[
\lambda_{\mathrm{ZZ}}(t)
=
\sum_{j=1}^d (l^{(j)}(t))_+.
\]
}

\revision{
The \emph{non-vectorized, unsigned version} applies Algorithm~\ref{alg:upper_bound}
once to the full Zig-Zag rate:
\[
\Lambda_{\mathrm{F,F}}(t)
=
B\left(\sum_{j=1}^d l^{(j)}_+\right)(t).
\]
Here, the $\mathrm{F,F}$ subscript indicates that the function uses neither the vectorized nor the signed strategy described below. 
}

\revision{
The \emph{vectorized, unsigned version} applies Algorithm~\ref{alg:upper_bound}
independently to each nonnegative coordinate rate and then sums the
resulting bounds:
\[
\Lambda_{\mathrm{T,F}}(t)
=
\sum_{j=1}^d B(l^{(j)}_+)(t).
\]
}

\revision{
Finally, the \emph{vectorized, signed version} constructs the bound
before applying the positive part. For each coordinate, Algorithm~\ref{alg:upper_bound} is
applied to the raw function $l^{(j)}$, and the resulting upper bounds are
combined as
\[
\Lambda_{\mathrm{T,T}}(t)
=
\sum_{j=1}^d
\left(B(l^{(j)})(t)\right)_+.
\]
This construction retains derivative information when $l^{(j)}$ crosses
zero, and $\Lambda_{\mathrm{T,F}}(t)\le \Lambda_{\mathrm{T,T}}(t)$. }

\revision{
The non-vectorized, unsigned strategy can be more sensitive to the grid
resolution because the positive part is taken separately for every
coordinate before summation, so the total rate may contain many
nondifferentiability points. Applying the bound construction directly to
this total rate can therefore miss local variations that are isolated when
the construction is applied coordinatewise.
}

\revision{
Note that we did not consider a non-vectorized, signed version corresponding to $\bigl(B(\sum_{j=1}^d l^{(j)})(t)\bigr)_+$ because it is not guaranteed to be an upper bound of $\lambda_{\mathrm{ZZ}}(t)$ even at the grid points. 
}



\subsubsection{Computing the next event}
Once the upper bound $\Lambda$ has been computed, the time of the next event can be simulated by integrating $\Lambda$. Given $e \sim \expd(1)$, the piecewise-constant form of the upper bound makes it easy to find $\tau$ such that $\int_0^\tau \Lambda(t)\dif t = e$. Let $i$ be the index such that $\int_0^{t_i} \Lambda(t)\dif t \leq e < \int_0^{t_{i+1}} \Lambda(t)\dif t$. Then, $\tau = t_i + \frac{e - \int_0^{t_i} \Lambda(t)\dif t}{\Lambda(t_i)}$. The procedure is presented in Algorithm \ref{alg:nextevent}.

\begin{algorithm}[htbp]
\footnotesize                                    
\caption{$\NextEvent$: computation of the time of the next event}
\label{alg:nextevent}
\begin{algorithmic}[1]                    
  \Require Upper bound $\Lambda$, number of segments $N$, horizon $\tmax$, random time $e$
  \Ensure Next-event time $\tau$

  \State $i \gets 0$
  \While{$i < N$ \textbf{and} $\displaystyle\int_{0}^{t_{i+1}} \!\!\Lambda(t)\,\mathrm{d}t \;\le e$}
        \State $i \gets i + 1$
  \EndWhile

  \If{$i = N$}                            \Comment{horizon reached}
        \State $\tau \gets \tmax + 1$     \Comment{sentinel value}
  \Else
        \State $\displaystyle
               \tau \gets t_i
                     + \frac{e-\int_{0}^{t_i} \Lambda(t)\,\mathrm{d}t}
                            {\Lambda(t_i)}$
  \EndIf

  \State \Return $\NextEvent(\Lambda,\tmax,N,e) \;\leftarrow\; \tau$
\end{algorithmic}
\end{algorithm}

\subsection{\revision{Adaptation of $\tmax$}}
\label{sec:adaptation_tmax_piecewise}

\revision{
The adaptation strategy for $\tmax$ presented in
Section~\ref{sec:adaptive_tmax} can also be used with the grid-based upper bound. However, the two update directions have asymmetric consequences in this setting. Reaching the horizon without an accepted event indicates that $\tmax$ may be unnecessarily short and incurs the cost of constructing a new bound. By contrast, a large value of $\tmax$ can produce a loose bound, increase the number of rejected candidate events, and make the grid-resolution condition more difficult to satisfy.
We therefore use two separate adaptation factors. Both $\alpha_+>1$ and $\alpha_->1$. After rejecting a candidate at time $\tauopt$, measured from the current bound's origin, we store $s\leftarrow\tauopt$ and immediately update $\tmax\leftarrow\tmax/\alpha_-$, while retaining the existing bound and its original grid. If the next candidate exceeds the shortened horizon or $s\ge\tmax$, we advance by $\max(s,\tmax)$, update $\tmax\leftarrow\alpha_+\tmax$, and restart with a new bound at that state. This prevents the horizon update from restarting before an already rejected candidate; see Algorithm~\ref{alg:automatic_pdmp}. We generally choose $1<\alpha_+<\alpha_-$, so that $\tmax$ is increased cautiously but
decreased more strongly when the current horizon produces inefficient thinning.
}

\revision{
To understand this adaptation mechanism, consider the updates on the
logarithmic scale. Let $N_{\mathrm{hit}}$ denote the number of times the
horizon is reached without an accepted event, including hits after rejections, and let
$N_{\mathrm{rej}}$ denote the number of rejected candidate events over a
given time window. Each horizon hit adds $\log(\alpha_+)$ to
$\log(\tmax)$, whereas each rejection subtracts
$\log(\alpha_-)$ from it. Over a window with no detected bound violations, the net change is therefore
\begin{align*}
\Delta\log(\tmax)
=
\log(\alpha_+)N_{\mathrm{hit}}
-
\log(\alpha_-)N_{\mathrm{rej}}.
\end{align*}
Each detected bound violation introduces an additional decrease of $\log 2$, as described in Section~\ref{sec:bound_errors}. When the adaptation is approximately stable and the contribution of these additional decreases is negligible, the net change remains close to zero on average, giving the heuristic balance
\[
\frac{N_{\mathrm{rej}}}{N_{\mathrm{hit}}}
\approx
\frac{\log(\alpha_+)}{\log(\alpha_-)}.
\]
This balance is examined numerically in
Appendix~\ref{app:sensitivity:tuning}.
}

\revision{
For fixed $\alpha_+$ and $\alpha_-$, a decrease in $N_{\mathrm{rej}}/N_{\mathrm{hit}}$ shifts the balance towards upward updates and larger horizons. Larger horizons in turn tend to reduce horizon hits, moving the updates back towards balance.
}

\revision{
To separate the asymmetry of the updates from their overall magnitude, we
parameterize $\alpha_+$ and $\alpha_-$ using two positive scalars
$(R,M)$:
\[
\log(\alpha_-)=\frac{M}{\sqrt{R}},
\qquad
\log(\alpha_+)=M\sqrt{R}.
\]
Under this parameterization,
\[
R
=
\frac{\log(\alpha_+)}{\log(\alpha_-)}
\]
controls the relative strength of the upward and downward updates. The parameter
\[
M
=
\sqrt{\log(\alpha_+)\log(\alpha_-)}
\]
controls the overall size of the updates without changing this ratio.
}

\revision{
In practice, $R$ should first be chosen according to the desired trade-off
between rejected candidate events and repeated bound constructions.
Choosing $R<1$, equivalently $\alpha_+<\alpha_-$, assigns a stronger
downward correction to a rejection. A smaller value of $R$ tends to reduce $\tmax$ and the number
of rejected candidate events, at the cost of more frequent horizon hits.
Conversely, increasing $R$ places relatively more weight on horizon hits,
leading to larger horizons and fewer bound recomputations, but potentially
more thinning rejections. For example, the balance above predicts about one
rejected candidate event for every four horizon hits when $R=0.25$.
}

\revision{
Once $R$ has been selected, choose a larger $M$ for faster adaptation or
a smaller $M$ to reduce oscillations of $\tmax$. The sensitivity to
$R$ and $M$ is examined in Appendix~\ref{app:trace:horizon}.
}

\subsection{Algorithm for the PDMP}
\revision{
Algorithm \ref{alg:automatic_pdmp} presents Algorithm \ref{alg:automatic_zigzag} adapted to use the upper bound $\Lambda$ constructed in Algorithm \ref{alg:upper_bound}.} It extends the method to general PDMPs defined by the triplet of flow, event rate, and jump kernel $(\phi_t,\lambda,Q)$. It also includes the adaptation of $\tmax$ presented in the previous section, parameterized by $\alpha_+$ and $\alpha_-$.
\begin{algorithm*}[htbp]                         
\footnotesize                                    
\caption{PDMP sampler with automatic tuning of $\tmax$ and piecewise-constant upper bound}
\label{alg:automatic_pdmp}
\begin{algorithmic}[1]                           
  \Require Initial state $(x_0,v_0)$; skeleton size $K$; rate $\lambda$
           ; initial horizon $\tmax$; grid size $N$
           ; flow map $\phi_t$; velocity kernel $Q$
           ; adaptation factors $\alpha_+$, $\alpha_-$
  \Ensure  Skeleton $\{t_k,x_k,v_k\}_{k=1}^{K}$

  \State $t\gets0,\,k\gets1,\,(x,v)\gets(x_0,v_0)$
  \While{$k\le K$}\label{alg:automatic_pdmp:outer}
     \State $H\gets\tmax$, $u\gets0$, $s\gets0$ \Comment{original grid horizon and last rejection time}
     \State $\Lambda\gets\UpperBound\!\bigl(\lambda(\,\cdot\,;x,v),\,H,N\bigr)$ \Comment{Alg.~\ref{alg:upper_bound}}
     \State Draw $e\sim\Exp(1)$
     \State $\tauopt\gets\NextEvent(\Lambda,H,N,e)$ \Comment{Alg.~\ref{alg:nextevent}}

     \While{$\tauopt\le\tmax$ \textbf{and} $s<\tmax$ \textbf{and} $u=0$}
        \State $r\gets\lambda(\tauopt;x,v)/\Lambda(\tauopt)$
        \If{$r>1$} \Comment{invalid bound}
           \State \revision{$\tmax\gets\tmax/2$}
           \State \revision{\textbf{continue} outer loop at line~\ref{alg:automatic_pdmp:outer}} \Comment{\revision{keep $(t,x,v)$ unchanged}}
        \Else
           \State Draw $u\sim\Ber(r)$
           \If{$u=1$} \Comment{event accepted}
              \State $t\gets t+\tauopt$;\; $(x,v)\gets\phi_{\tauopt}(x,v)$
              \State Draw $v\sim Q(x,v,\cdot)$
              \State Save $(t,x,v)$ into $(t_k,x_k,v_k)$;\; $k\gets k+1$
           \Else \Comment{candidate rejected}
              \State $s\gets\tauopt$;\; $\tmax\gets\tmax/\alpha_-$
              \State Draw $e'\sim\Exp(1)$;\; $e\gets e+e'$
              \State $\tauopt\gets\NextEvent(\Lambda,H,N,e)$
           \EndIf
        \EndIf
     \EndWhile

     \If{$u=0$} \Comment{window ends without an accepted event}
        \State $\tauopt\gets\max(s,\tmax)$
        \State $t\gets t+\tauopt$;\; $(x,v)\gets\phi_{\tauopt}(x,v)$
        \State $\tmax\gets\alpha_+\tmax$
     \EndIf
  \EndWhile

  \State \Return $\bigl\{\,t_k,x_k,v_k\,\bigr\}_{k=1}^{K}$
\end{algorithmic}
\end{algorithm*}

\subsection{Dealing with errors in the upper bound}
\label{sec:bound_errors}

\revision{
To monitor possible failures of the bound, we compare $\lambda(\tau)$ and
$\Lambda(\tau)$ at each candidate time $\tau$. The inequality
$\lambda(\tau)>\Lambda(\tau)$ proves that the bound is invalid at $\tau$,
but its absence does not rule out violations
between evaluated times. Upon detection in adaptive mode, we halve $\tmax$ and restart the outer loop at the unchanged flow origin to recompute the bound. This restart skips the horizon-advance step. In fixed-horizon experiments, the horizon remains fixed and the thinning ratio is capped at $1$ for the acceptance decision, while the uncapped ratio is retained for diagnostics. Neither procedure retrospectively corrects the simulated path; exactness remains guaranteed
only under the conditions of Theorem~\ref{theorem:exactness}.
}

\revision{
As a complementary diagnostic, we also consider an integral consistency
check. Suppose that along a deterministic trajectory, $\lambda(t)\ge l(t)$, where
\[
l(t)
=
\frac{\mathrm{d}}{\mathrm{d}t}\Phi(x_t)
\]
for an appropriate potential $\Phi$. Then
\[
\int_0^{\tmax}l(t)\,\mathrm{d}t
=
\Phi(x_{\tmax})-\Phi(x_0).
\]
Consequently, any valid upper bound $\Lambda$ must satisfy the necessary condition
\[
\int_0^{\tmax}\Lambda(t)\,\mathrm{d}t
\geq
\Phi(x_{\tmax})-\Phi(x_0).
\]
A violation of this inequality proves that the constructed bound is
invalid. However, satisfying the inequality does not certify pointwise
domination: local violations may remain undetected because the check only
compares integrated quantities. We therefore use this condition as a
post-hoc sanity check rather than as part of the event-time simulation.
Implementation details and numerical results are reported in
Appendix~\ref{app:integral_sanity_check}.
}

\section{Numerical experiments}
\label{sec:numerical_experiments}

To illustrate the benefits of this new method, we compare it with the automatic Zig-Zag algorithm in two examples. The algorithms are implemented in Python using the JAX library \cite{jax2018github}. The implementation is available on GitHub at \url{https://github.com/charlyandral/pdmp_jax} and on PyPI at \url{https://pypi.org/project/pdmp-jax}. \revision{All experiments, including the wall-clock time comparisons reported in Figure~\ref{fig:time_mixture_gaussian}, were run on the same MacBook Pro equipped with an Apple M5 chip and 24 GB of memory, running macOS 27.0. The implementations used Python 3.13.15 and JAX 0.6.2. All computations used the CPU backend, with independent runs parallelized across up to eight workers.}

\subsection{Mixture of two Gaussian distributions with different scales}

For the first example, we consider a mixture of two Gaussian distributions in dimension 2. The target distribution is $\Pi = \frac{1}{2}\mathcal{N}((0,0), I_2) + \frac{1}{2} \mathcal{N}((1,1), \sigma^2 I_2)$, with $\sigma = 0.03$. Although the two modes are close, this target can be difficult to sample because the second mode is sharply peaked. The difference in scale between the two modes makes it difficult to choose a suitable horizon $\tmax$ for the entire space. If $\tmax$ is too large, the narrow mode can easily be missed; if it is on the scale of $\sigma$, the process moves slowly through the broad mode. This issue can be diagnosed by examining the sample mean. If the second mode is not visited, the mean of each component is smaller than the true mean $(0.5,0.5)$.

For this example, we use the Bouncy Particle Sampler with a refreshment rate of $0.1$. We first compare the effects of the horizon $\tmax$ and the number of grid points on the sample mean. The two adaptation parameters are set to $\alpha_+ = 1.01$ and $\alpha_- = 1.04$. \revision{We use the signed version from Section \ref{sec:bps_rate} to compute the upper bound.}

We compare four values of $\tmax$: $0$, $0.01$, $0.1$, and $1$. Here, $0$ denotes an adaptive horizon, whereas the other three values denote fixed horizons. We also consider six values of the number of grid points $N$: $0$, $5$, $10$, $20$, $50$, and $100$. Here, $N=0$ means that the upper bound is computed using Brent's algorithm, as in the automatic Zig-Zag algorithm. \revision{In this case, only the default unsigned version is used.} We run $10$ independent processes with $10^6$ skeleton points for each combination of $\tmax$ and $N$. The mean of each process is computed analytically from its skeleton points. The results are presented in Figure \ref{fig:mean_mixture_gaussian}.

\begin{figure}[h]
    \centering
    \includegraphics[width=\linewidth]{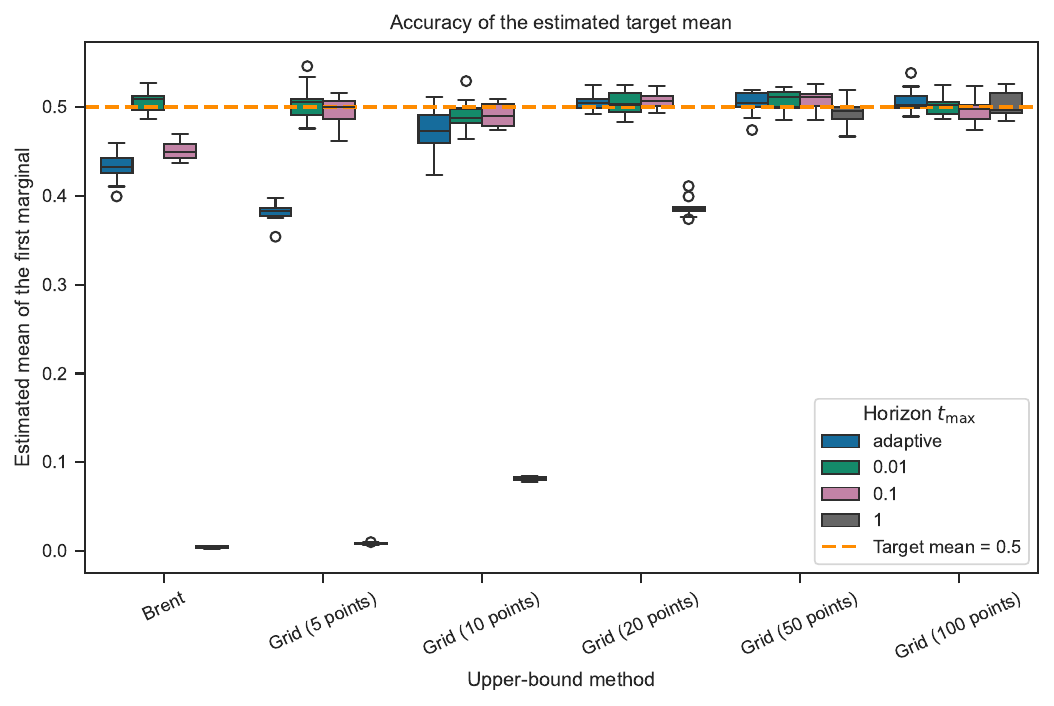}
    \caption{\revision{Estimated mean of the first marginal for the mixture of Gaussian distributions, shown by horizon setting and upper-bound method. The orange dashed line marks the target value $0.5$.}}
    \label{fig:mean_mixture_gaussian}
\end{figure}

For a coarse grid or an optimization-based bound, the mean of the process is far from the true mean unless the horizon is very short. For instance, using Brent's algorithm with $\tmax=1$ produces a mean close to zero across all 10 runs, indicating that the narrow mode is missed every time. \revision{For Brent's algorithm, horizon adaptation also fails to produce an unbiased result.} However, when the grid contains sufficiently many points (here, $>20$), the sample mean is close to its true value.

Finally, we examine the computation time for the different values of $\tmax$ and $N$. The results are presented in Figure \ref{fig:time_mixture_gaussian}. \revision{Thus, although $\tmax=0.01$ gives an estimated mean close to
the target value in Figure \ref{fig:mean_mixture_gaussian},
this accuracy comes at a higher computational cost}. This qualifies the impression from Figure \ref{fig:mean_mixture_gaussian} that $\tmax=0.01$ is a good choice. When $N$ is small, $\tmax$ must also be small to capture the narrow mode; however, this makes the algorithm inefficient while it moves through the broad mode.

The figure also shows that \revision{with an adaptive horizon, using Brent's algorithm is slower than using a grid}, even one with $100$ points.
This is because evaluating the rate function on a grid is much faster than using an optimization algorithm that evaluates points sequentially. Finally, adapting $\tmax$ with a grid is very efficient in this case and is among the fastest configurations considered. This is not the case for Brent's algorithm, for which adaptation is slower than using a fixed $\tmax$ \revision{of $0.1$ or $1$}. We have not investigated the reason for this behavior.

Overall, choosing $N=50$ (or even $N=20$) and adapting $\tmax$ gives the best results in terms of computation time and sample accuracy: it produces an accurate sample in approximately \revision{$19\%$ (or $17\%$ for $N=20$)} of the time required by the \revision{Brent algorithm} with $\tmax=0.01$. \revision{This favorable balance is also illustrated in Figure~\ref{fig:time_vs_rmse_mixture_gaussian}, which plots RMSE against mean wall-clock time.}

\begin{figure}[h]
    \centering
    \includegraphics[width=\linewidth]{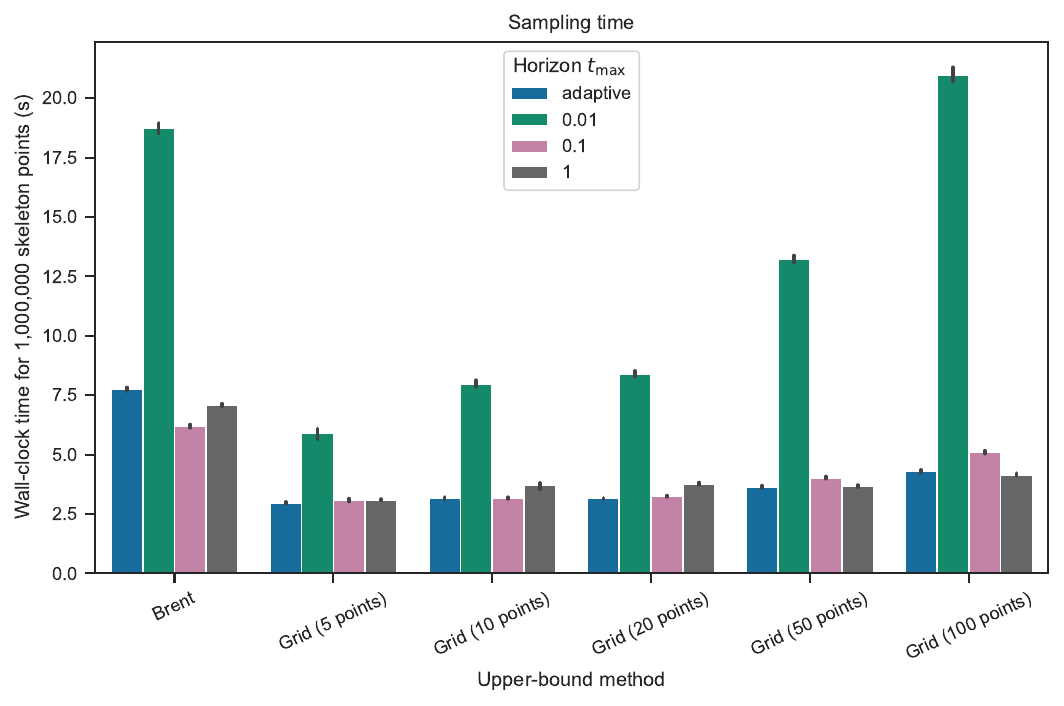}
    \caption{\revision{Mean wall-clock time required to simulate $10^6$ skeleton points for the mixture of two Gaussian distributions, shown by horizon setting and upper-bound method.}}
    \label{fig:time_mixture_gaussian}
\end{figure}




\begin{figure}[h]
    \centering
    \includegraphics[width=\linewidth]{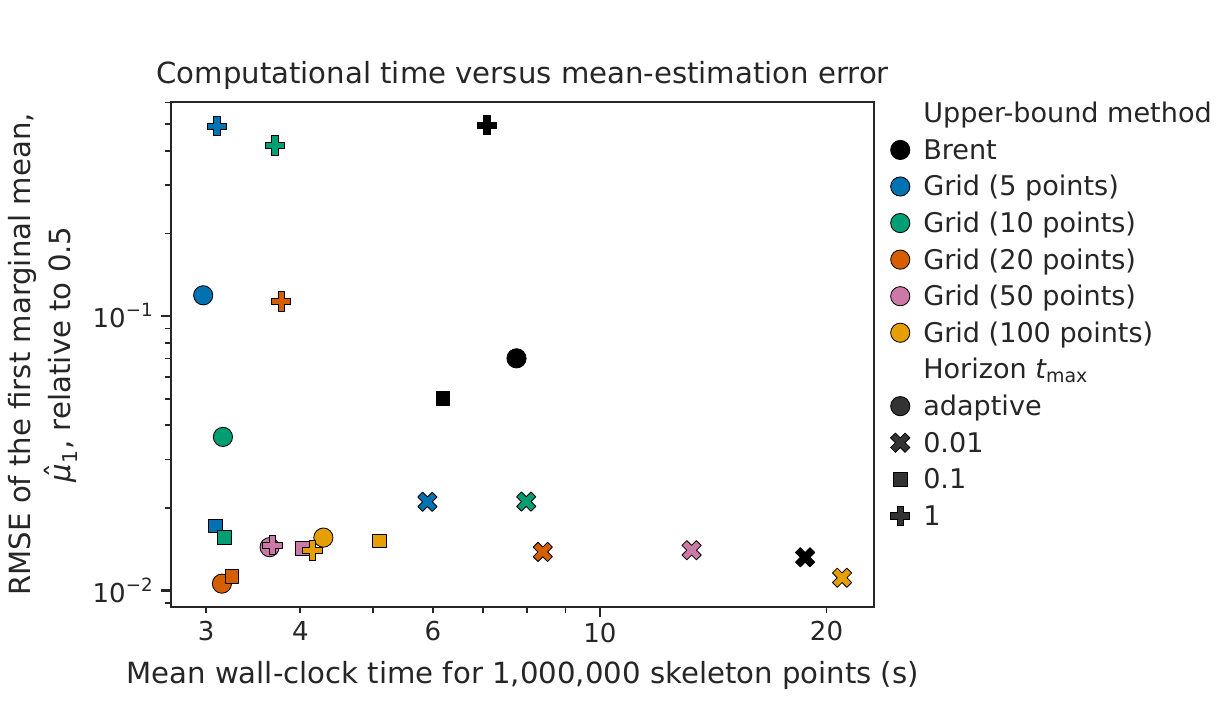}
    \caption{\revision{Computational time versus estimation error for the mixture of two Gaussian distributions. Each point represents one combination of horizon setting and upper-bound method. The horizontal axis is the mean wall-clock time required to simulate $10^6$ skeleton points, and the vertical axis is the RMSE of the estimated first marginal mean $\hat{\mu}_1$ relative to the true value $0.5$, computed across $10$ independent runs. Both axes use logarithmic scales. Colors indicate the upper-bound method and marker shapes indicate the horizon setting.}}
    \label{fig:time_vs_rmse_mixture_gaussian}
\end{figure}

\subsection{Local mixture of 20 Gaussian distributions}

For the second example, we consider a mixture of $20$ Gaussian distributions in dimension $2$. It is constructed as follows: for each $i=1,\dots,20$, we sample a mean $\mu_i$ from a Gaussian distribution in $\Rset^2$ with variance $3^2$. The target distribution is then $\Pi = \frac{1}{20}\sum_{i=1}^{20} \mathcal{N}(\mu_i, I_2)$. This construction produces a nonconvex potential that can be difficult to bound. The means are sampled once at the beginning of the experiment and held fixed across runs.

We use this target to evaluate the errors in the upper bound, and to illustrate the effect of the computational strategies described in Sections \ref{sec:bps_rate} and \ref{sec:zigzag_rate}.

\subsubsection{Results for the Zig-Zag sampler}

\revision{We first consider the Zig-Zag sampler. The horizon $\tmax$ is adaptive, with $\alpha_+ = 1.01$ and $\alpha_- = 1.04$. We compare results for different grid sizes $N$ ($5, 10, 20, 50$) and for Brent's algorithm, which is displayed as a separate upper-bound method. For the grid-based bounds, we compare the non-vectorized, unsigned; vectorized, unsigned; and vectorized, signed versions described in Section~\ref{sec:zigzag_rate}.}

\revision{
Following Section~\ref{sec:bound_errors}, we record a bound violation whenever
$\lambda(\tau)/\Lambda(\tau)>1$ at a proposed candidate time $\tau$.
Ratios above one indicate bound errors and cannot be interpreted as acceptance
probabilities. We report the number of detected errors and the mean excess ratio
$\lambda(\tau)/\Lambda(\tau)-1$ conditional on $\lambda(\tau)/\Lambda(\tau)>1$.
}

\begin{figure*}[h]
    \centering
    \includegraphics[width=.82\textwidth]{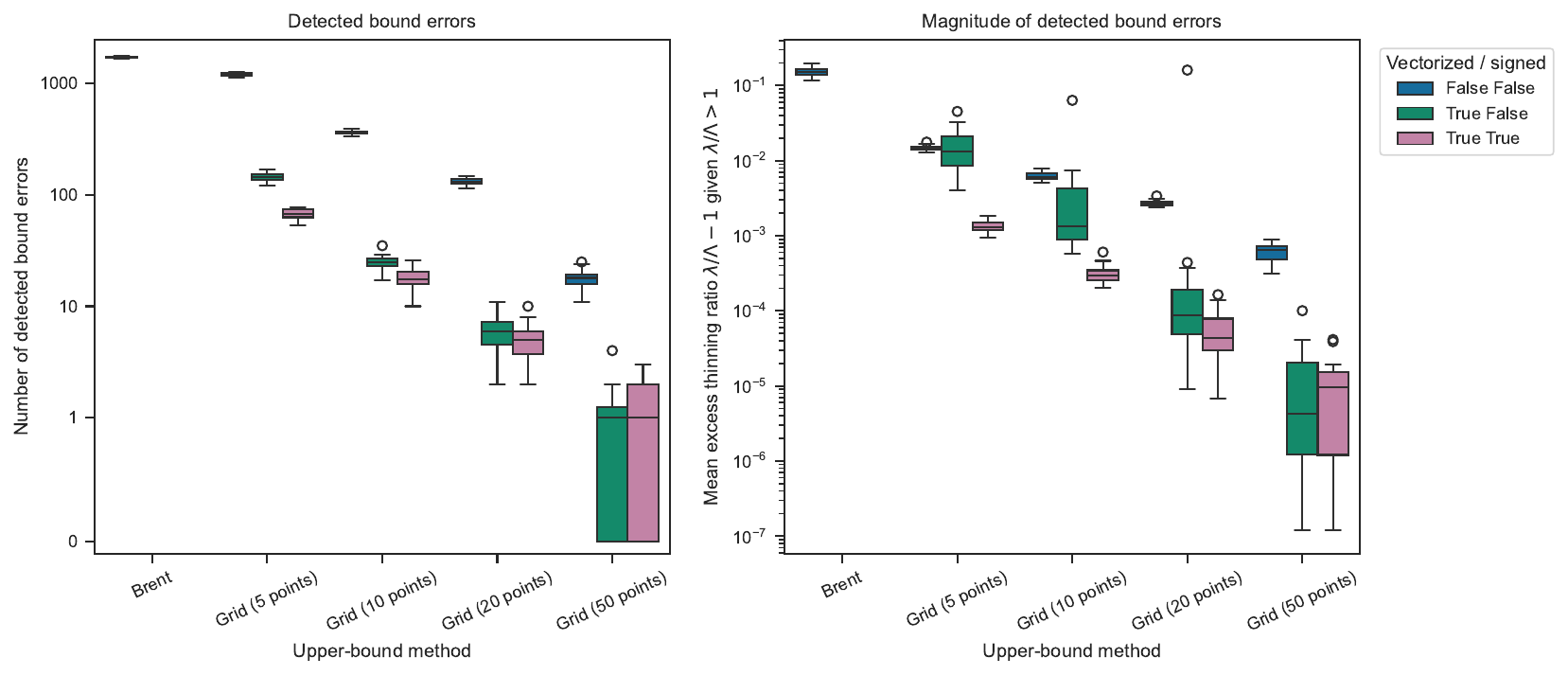}
    \caption{\revision{Detected bound errors in the Zig-Zag sampler for the mixture of
    20 Gaussian distributions. Left: number of candidate times $\tau$ at which
    the thinning ratio $\lambda(\tau)/\Lambda(\tau)$ exceeds $1$. Right: mean
    excess thinning ratio $\lambda(\tau)/\Lambda(\tau)-1$, conditional on
    $\lambda(\tau)/\Lambda(\tau)>1$ and computed only for runs containing at
    least one detected error. For Brent
    optimization, only the non-vectorized, unsigned version was run.}}
    \label{fig:zigzag_mixture_gaussian}
\end{figure*}

\revision{
The results show that increasing the number of grid points generally reduces both
the number and the magnitude of bound violations. Among the grid-based
constructions, the non-vectorized, unsigned version produces the
largest number of violations, while the vectorized, signed version is
the most robust of the three versions in this experiment. This behavior
is consistent with the discussion of nondifferentiability in
Section~\ref{sec:zigzag_rate}.
}

\subsubsection{Results for the Boomerang sampler}

\revision{We repeat the experiment for the Boomerang sampler with a refreshment rate of $0.1$. We compare the unsigned and signed versions described in Section~\ref{sec:bps_rate}.}

\revision{The results are presented in Figure \ref{fig:boomerang_mixture_gaussian}. Both the unsigned and signed versions produce detected upper-bound errors in this nonconvex example. The two versions are distinguished by color in both panels; the right panel includes only runs containing at least one detected error.}

\begin{figure*}[h]
    \centering
    \includegraphics[width=.82\textwidth]{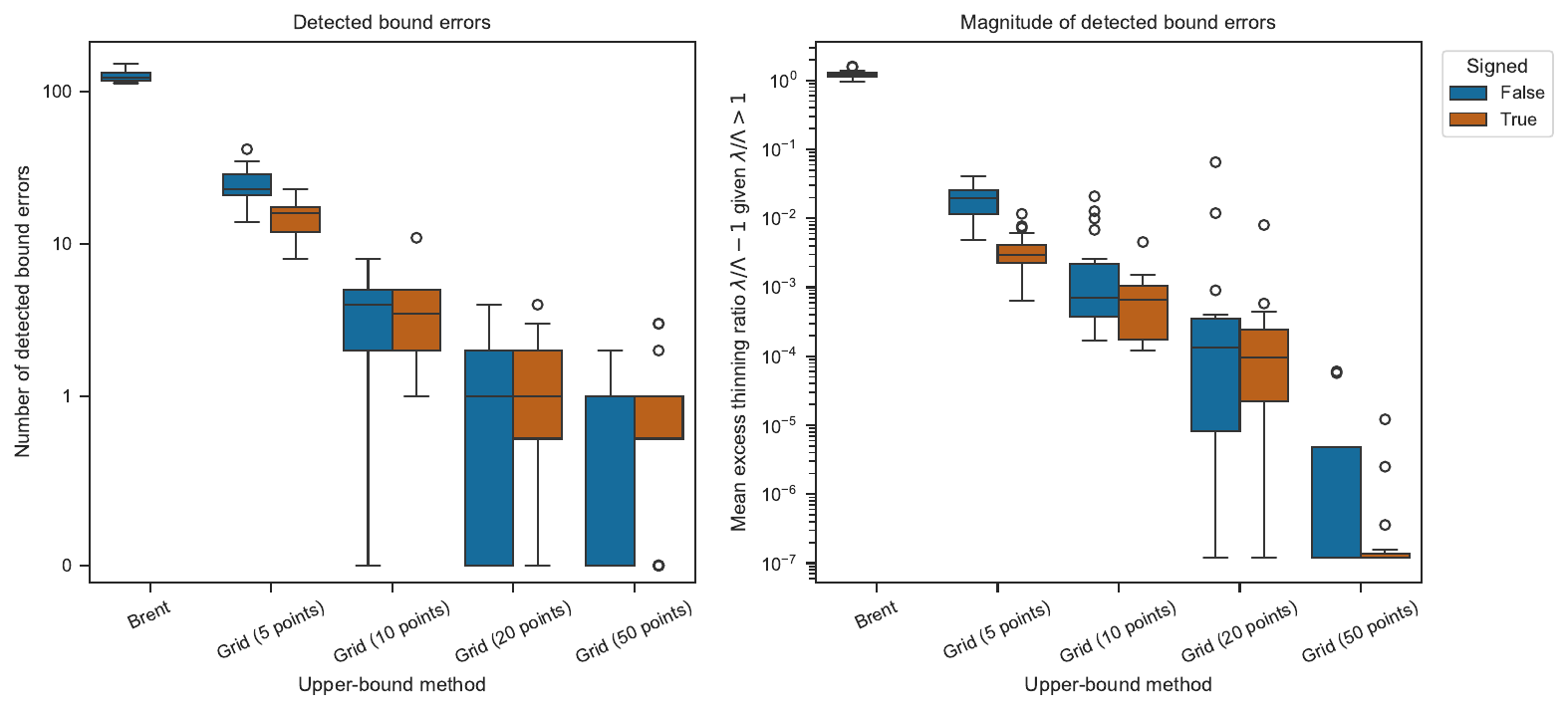}
    \caption{\revision{Detected bound errors in the Boomerang sampler for the mixture of
    Gaussian distributions. Left: number of candidate times $\tau$ at which
    the thinning ratio $\lambda(\tau)/\Lambda(\tau)$ exceeds $1$. Right: mean
    excess thinning ratio $\lambda(\tau)/\Lambda(\tau)-1$, conditional on
    $\lambda(\tau)/\Lambda(\tau)>1$ and computed only for runs containing at
    least one detected error. Colors distinguish the unsigned and signed versions.
    For Brent optimization, only the unsigned version was run.}}
    \label{fig:boomerang_mixture_gaussian}
\end{figure*}

\subsection{High-dimensional setting}

Finally, we consider a high-dimensional setting. The target is a banana distribution in the first two dimensions and a Gaussian distribution in the remaining dimensions. The potential is defined as \revision{$U(x) = \frac{1}{2}\bigl[x_1^2 + (x_2 - x_1^2 + 1)^2 + \sum_{i=3}^d x_i^2\bigr]$}. We run the Forward Ref version of the Forward Event Chain algorithm \cite{michelForwardEventChainMonte2020}, with the modification that the orthogonal rotation is performed with probability $0.1$ at each event rather than according to a deterministic schedule, and with no refreshment rate ($\bar{\lambda} = 0$). The horizon $\tmax$ is again adaptive, with $\alpha_+ = 1.01$ and $\alpha_- = 1.04$. The process is run in dimension $d=300$ for $10^5$ skeleton points, with $20$ independent runs for each combination of $N \in \{0,3,5,10,20\}$ and signed-bound setting \revision{(except for $N=0$, for which only the unsigned version is used)}.

Because the target is not multimodal, upper-bound errors are rare: the mean number of errors per chain is \revision{$0.40$ for $N=0$ (Brent's algorithm). In all other cases, no upper-bound errors occur.} Table \ref{tab:high_dim} presents several quantities of interest for the different signed-bound settings and grid sizes, averaged over the $20$ runs. We report the running time in seconds, the average value of $\tmax$, the mean thinning acceptance rate, the number of rejections, and the number of times the horizon is reached without an event, \revision{including hits after rejections in Algorithm~\ref{alg:automatic_pdmp}}.

Because the potential is nearly convex, the strategies produce very similar results in this case. Nevertheless, the experiment illustrates how the algorithm behaves as the number of grid points $N$ changes. Increasing the number of grid points improves the accuracy of the upper bound, resulting in a higher thinning acceptance rate and fewer rejections.

\revision{
For $\alpha_+=1.01$ and $\alpha_-=1.04$,
$R=\log(\alpha_+)/\log(\alpha_-)\approx0.25370$. The ratio of the mean rejection
count to the mean horizon-hit count, computed from the unrounded data
underlying Table~\ref{tab:high_dim}, ranges from
$0.25301$ to $0.25370$ across the nine configurations. These ratios differ
from $R$ by less than $0.3\%$ in relative terms and are consistent with the
balance in Section~\ref{sec:adaptation_tmax_piecewise}.
}

\begin{table}[tb]
\centering
\scriptsize
\setlength{\tabcolsep}{3pt}
\caption{300-d banana target.
Each run simulates $10^{5}$ skeleton points.
\revision{“sgn.” indicates whether the signed-bound version is used.}
\revision{AR is the mean of $10^5/(10^5+N_{\mathrm{rej}})$ across runs.}
Columns “rejects” and “horizon” are counts divided by $10^{3}$.}
\label{tab:high_dim}
\begin{tabular}{@{}crrrrrr@{}}
\toprule
$N$ & \revision{sgn.} & time (s) & $\overline{\tmax}$ & AR &
rejects & horizon \\
\midrule
0 & \revision{off} & \revision{3.15} & \revision{0.82} & \revision{0.61} & \revision{65.0} & \revision{256.1} \\
3 & \revision{off} & \revision{2.39} & \revision{1.15} & \revision{0.70} & \revision{42.8} & \revision{168.6} \\
3 & \revision{on} & \revision{2.46} & \revision{1.14} & \revision{0.70} & \revision{42.8} & \revision{168.7} \\
5 & \revision{off} & \revision{2.53} & \revision{1.57} & \revision{0.78} & \revision{27.9} & \revision{109.9} \\
5 & \revision{on} & \revision{2.79} & \revision{1.57} & \revision{0.78} & \revision{27.9} & \revision{109.9} \\
10 & \revision{off} & \revision{2.30} & \revision{2.17} & \revision{0.86} & \revision{16.4} & \revision{64.7} \\
10 & \revision{on} & \revision{2.74} & \revision{2.17} & \revision{0.86} & \revision{16.4} & \revision{64.7} \\
20 & \revision{off} & \revision{3.48} & \revision{2.74} & \revision{0.91} & \revision{9.5} & \revision{37.6} \\
20 & \revision{on} & \revision{3.51} & \revision{2.74} & \revision{0.91} & \revision{9.5} & \revision{37.7} \\
\bottomrule
\end{tabular}
\end{table}

\section{Conclusion}
Building on the automatic Zig-Zag algorithm \cite{corbellaAutomaticZigZagSampling2022}, we have proposed a new method for computing an upper bound on the rate function of a PDMP. This method uses a time grid to create a piecewise-constant bound. We have shown that this bound is valid under conditions that depend on the critical points of the rate function and the spacing between grid points. We also modified the algorithm to adapt the horizon $\tmax$ automatically.
We used the method on different PDMP samplers, and showed that it is more efficient than the automatic Zig-Zag algorithm in terms of errors in the upper bound, and in terms of computation time. \revision{Finally, the accompanying JAX implementation is designed to be easy to use and extend, and a new PDMP can be specified in a few lines of code.} Further work could be done to create a more efficient adaptation of the horizon $\tmax$, in particular for targets with a complex geometry  and advanced PDMPs that incorporate geometric information.

\appendix
\section{Log ratio of the tuning parameters}
\label{app:sensitivity:tuning}

We report the empirical ratio $N_{\mathrm{rej}}/N_{\mathrm{hit}}$ between the number of rejected proposals and the number of horizon hits without any events, and the ratio of the logarithms of multiplicative factors $R=\log\alpha_+/\log\alpha_-$. We consider the Bouncy Particle sampler, dimension $d=3$.

\revision{The comparison with $y=x$ in Figure~\ref{fig:appendix:sensitivity:alphas} assesses the balance $N_{\mathrm{rej}}/N_{\mathrm{hit}}\approx R$ in Section~\ref{sec:adaptation_tmax_piecewise}.}

\begin{figure}[!htbp]
    \centering
    \includegraphics[width=0.4\textwidth]{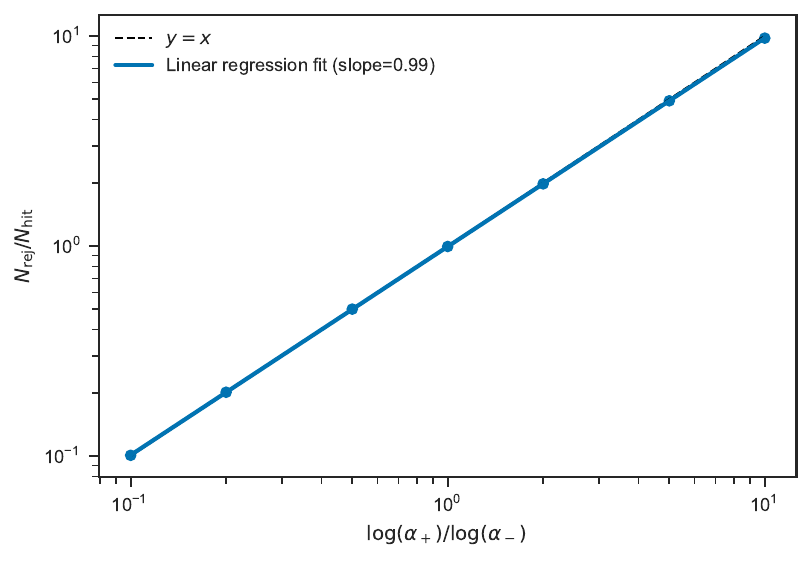}
    \caption{
    \revision{The horizontal axis shows $R$,
and the vertical axis shows $N_{\mathrm{rej}}/N_{\mathrm{hit}}$.} Both axes use logarithmic scales. Each point corresponds to one tuning setting (aggregated over $10$ runs of $10^4$ skeleton points). The dashed line is $y=x$; the solid line is a linear regression fit.}
    \label{fig:appendix:sensitivity:alphas}
\end{figure}

\section{Trace plots for the adaptive horizon}
\label{app:trace:horizon}

We complement Appendix~\ref{app:sensitivity:tuning} by adding trace plots of the horizon process $\tmax$ for a multimodal target.
We run the Zig-Zag sampler on a mixture of $20$ Gaussians in $d=2$, using grid size $N=20$ and initial horizon $\tmax=3.5$.
\revision{As shown in Figure~\ref{fig:appendix:horizon:trace},} a larger $M$ yields faster adaptation but also greater variability, whereas $R$ controls the overall magnitude of $\tmax$.

\begin{figure*}[!htbp]
    \centering
    \includegraphics[width=0.82\textwidth]{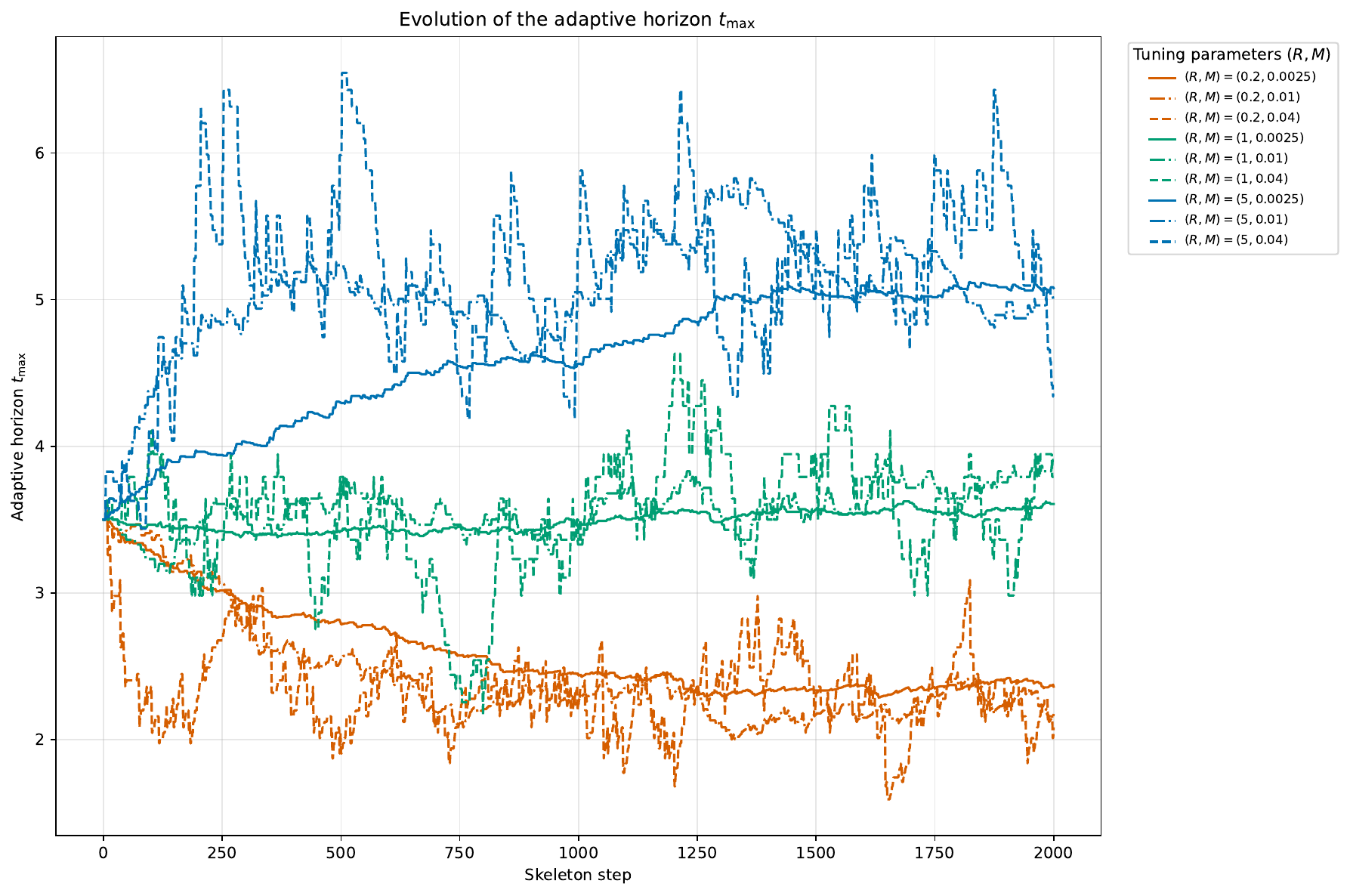}
    \caption{\revision{Adaptive-horizon traces for different tuning settings $(R,M)$. The legend reports only $(R,M)$ to avoid overcrowding. For $M=0.0025,0.01,0.04$, respectively, the resulting $(\alpha_-,\alpha_+)$ values are $(1.0056,1.0011)$, $(1.0226,1.0045)$, and $(1.0936,1.0180)$ when $R=0.2$; $(1.0025,1.0025)$, $(1.0101,1.0101)$, and $(1.0408,1.0408)$ when $R=1$; and $(1.0011,1.0056)$, $(1.0045,1.0226)$, and $(1.0180,1.0936)$ when $R=5$.}}
    \label{fig:appendix:horizon:trace}
\end{figure*}

\section{Sensitivity to grid size and horizon}
\label{app:sensitivity:grid:tmax}

Finally, we illustrate how accuracy and bound-violation frequency vary with the grid size $N$ and the horizon parameter $\tmax$ on a simple bimodal target.
The experiment uses the Bouncy Particle sampler on a two-component Gaussian mixture in $d=2$.
Here, $\mu=(0.5,0.5)$ denotes the mean of the target distribution $\Pi$ and $\hat\mu$ is the empirical mean estimated from a skeleton obtained by the sampler. 

As shown in Figure~\ref{fig:appendix:sensitivity:grid:tmax}, a larger $N$ \revision{generally reduces the mean-estimation error $\|\hat\mu-\mu\|_2$, although the detected upper-bound violation rate need not decrease monotonically}. Meanwhile, a smaller grid size combined with a larger $\tmax$ leads to greater errors, reflecting the difficulty of bounding the rate function over a longer horizon with a coarse grid. The adaptive-horizon setting ($\tmax=0$) \revision{performs well for sufficiently fine grids, but still gives a substantial mean-estimation error for $N=5$}.

Note that upper-bound violations are measured only at \revision{candidate times}, so violations occurring between \revision{evaluated candidate times} may be missed. This likely explains the relatively small number of detections for $\tmax=1.0$ when using fewer grid points (coarser grids).

\begin{figure*}[!htbp]
    \centering
    \includegraphics[width=0.9\textwidth]{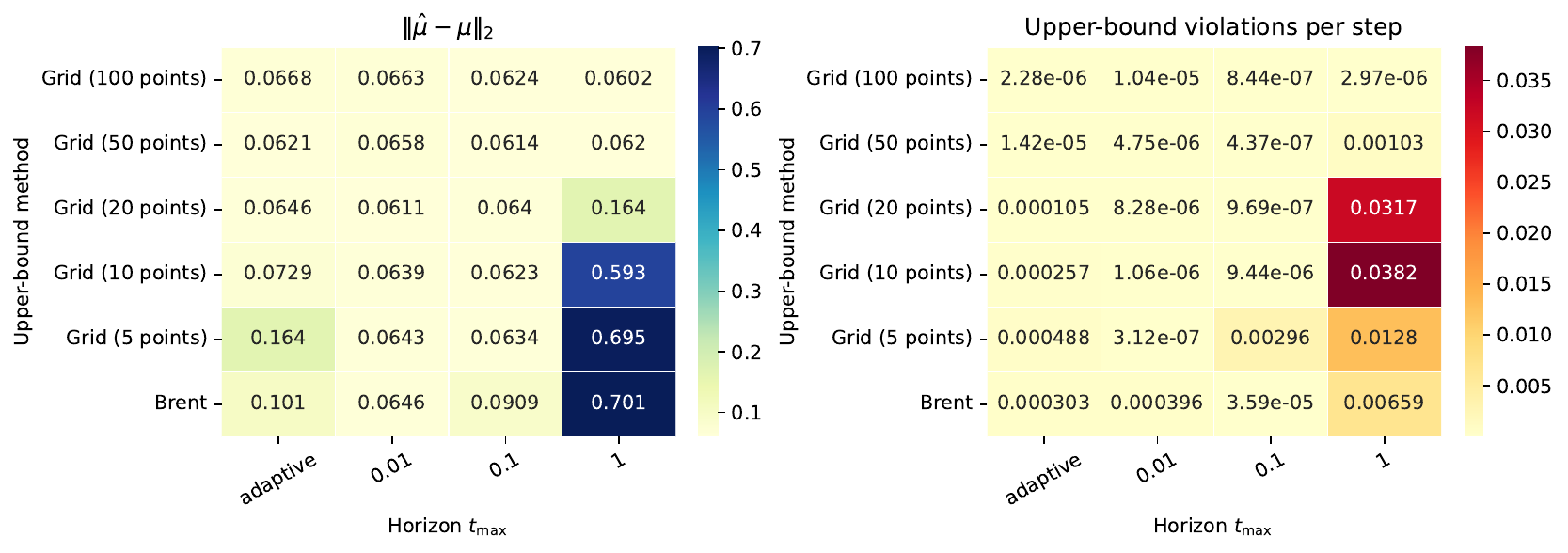}
    \caption{\revision{Sensitivity to the upper-bound method and horizon $\tmax$. The adaptive horizon and Brent optimization are displayed explicitly rather than as zero-valued settings. Left: average mean-estimation error $\|\hat\mu-\mu\|_2$ across seeds. Right: fraction of skeleton steps containing at least one detected upper-bound violation, averaged over runs.}}
    \label{fig:appendix:sensitivity:grid:tmax}
\end{figure*}

\section{Integral sanity check}
\label{app:integral_sanity_check}

\revision{
We performed a supplementary post-hoc integral diagnostic using the
20-component Gaussian-mixture target from
Section~\ref{sec:numerical_experiments} and the horizon-update rule in
Algorithm~\ref{alg:automatic_pdmp}. We considered the same 22
combinations of sampler, upper-bound method, and vectorized/signed setting as in
Figures~\ref{fig:zigzag_mixture_gaussian}
and~\ref{fig:boomerang_mixture_gaussian}, using three independent runs of
$10^4$ skeleton points for each combination. The diagnostic was applied
whenever a bound was constructed, including after an adaptive-horizon
update, with numerical tolerance $\varepsilon_{\mathrm{int}}=10^{-6}$.
For Zig-Zag, the potential difference was computed from the negative
log-target density; for Boomerang, it was computed from the residual
potential obtained after removing the Gaussian reference potential.
}

\revision{
In total, we checked $1{,}121{,}752$ bound constructions:
$648{,}599$ for Zig-Zag and $473{,}153$ for Boomerang. The pointwise
diagnostic detected 118 violations, comprising 113 for Zig-Zag and 5 for
Boomerang. The integral condition was violated once, for Zig-Zag
using Brent optimization with the non-vectorized, unsigned version.
In this run, the minimum integral slack (the left-hand side minus
the right-hand side of the integral condition) was
$-4.97317\times10^{-3}$.
The integral violation occurred at a skeleton step already flagged by the
pointwise diagnostic. No integral violation was observed for Boomerang.
Thus, the integral diagnostic did not identify any additional affected skeleton steps in these runs. This
diagnostic was run separately from the timing experiments and did not
affect the reported computational times.
}

\backmatter

\bmhead{Acknowledgements}
Part of this work was done while Charly Andral was part of the International Internship Program of the Institute of Statistical Mathematics, Tokyo, Japan.

Kengo Kamatani is supported by JST CREST Grant Number JPMJCR2115 and JSPS KAKENHI Grant Numbers JP21K18589 and JP20H04149. 

\bibliographystyle{sn-mathphys-num}
\bibliography{automatic_pdmp}


\begin{thebibliography}{29}
\ifx \bisbn   \undefined \def \bisbn  #1{ISBN #1}\fi
\ifx \binits  \undefined \def \binits#1{#1}\fi
\ifx \bauthor  \undefined \def \bauthor#1{#1}\fi
\ifx \batitle  \undefined \def \batitle#1{#1}\fi
\ifx \bjtitle  \undefined \def \bjtitle#1{#1}\fi
\ifx \bvolume  \undefined \def \bvolume#1{\textbf{#1}}\fi
\ifx \byear  \undefined \def \byear#1{#1}\fi
\ifx \bissue  \undefined \def \bissue#1{#1}\fi
\ifx \bfpage  \undefined \def \bfpage#1{#1}\fi
\ifx \blpage  \undefined \def \blpage #1{#1}\fi
\ifx \burl  \undefined \def \burl#1{\textsf{#1}}\fi
\ifx \doiurl  \undefined \def \doiurl#1{\url{https://doi.org/#1}}\fi
\ifx \betal  \undefined \def \betal{\textit{et al.}}\fi
\ifx \binstitute  \undefined \def \binstitute#1{#1}\fi
\ifx \binstitutionaled  \undefined \def \binstitutionaled#1{#1}\fi
\ifx \bctitle  \undefined \def \bctitle#1{#1}\fi
\ifx \beditor  \undefined \def \beditor#1{#1}\fi
\ifx \bpublisher  \undefined \def \bpublisher#1{#1}\fi
\ifx \bbtitle  \undefined \def \bbtitle#1{#1}\fi
\ifx \bedition  \undefined \def \bedition#1{#1}\fi
\ifx \bseriesno  \undefined \def \bseriesno#1{#1}\fi
\ifx \blocation  \undefined \def \blocation#1{#1}\fi
\ifx \bsertitle  \undefined \def \bsertitle#1{#1}\fi
\ifx \bsnm \undefined \def \bsnm#1{#1}\fi
\ifx \bsuffix \undefined \def \bsuffix#1{#1}\fi
\ifx \bparticle \undefined \def \bparticle#1{#1}\fi
\ifx \barticle \undefined \def \barticle#1{#1}\fi
\bibcommenthead
\ifx \bconfdate \undefined \def \bconfdate #1{#1}\fi
\ifx \botherref \undefined \def \botherref #1{#1}\fi
\ifx \url \undefined \def \url#1{\textsf{#1}}\fi
\ifx \bchapter \undefined \def \bchapter#1{#1}\fi
\ifx \bbook \undefined \def \bbook#1{#1}\fi
\ifx \bcomment \undefined \def \bcomment#1{#1}\fi
\ifx \oauthor \undefined \def \oauthor#1{#1}\fi
\ifx \citeauthoryear \undefined \def \citeauthoryear#1{#1}\fi
\ifx \endbibitem  \undefined \def \endbibitem {}\fi
\ifx \bconflocation  \undefined \def \bconflocation#1{#1}\fi
\ifx \arxivurl  \undefined \def \arxivurl#1{\textsf{#1}}\fi
\csname PreBibitemsHook\endcsname

\bibitem[\protect\citeauthoryear{Corbella
  et~al.}{2022}]{corbellaAutomaticZigZagSampling2022}
\begin{barticle}
\bauthor{\bsnm{Corbella}, \binits{A.}},
\bauthor{\bsnm{Spencer}, \binits{S.E.F.}},
\bauthor{\bsnm{Roberts}, \binits{G.O.}}:
\batitle{Automatic {{Zig-Zag}} sampling in practice}.
\bjtitle{Statistics and Computing}
\bvolume{32}(\bissue{6}),
\bfpage{107}
(\byear{2022})
\doiurl{10.1007/s11222-022-10142-x}
\end{barticle}
\endbibitem

\bibitem[\protect\citeauthoryear{Horowitz}{1991}]{HOROWITZ1991247}
\begin{barticle}
\bauthor{\bsnm{Horowitz}, \binits{A.M.}}:
\batitle{A generalized guided monte carlo algorithm}.
\bjtitle{Physics Letters B}
\bvolume{268}(\bissue{2}),
\bfpage{247}--\blpage{252}
(\byear{1991})
\doiurl{10.1016/0370-2693(91)90812-5}
\end{barticle}
\endbibitem

\bibitem[\protect\citeauthoryear{Diaconis
  et~al.}{2000}]{diaconisAnalysisNonreversibleMarkov2000}
\begin{barticle}
\bauthor{\bsnm{Diaconis}, \binits{P.}},
\bauthor{\bsnm{Holmes}, \binits{S.}},
\bauthor{\bsnm{Neal}, \binits{R.M.}}:
\batitle{Analysis of a {{Nonreversible Markov Chain Sampler}}}.
\bjtitle{The Annals of Applied Probability}
\bvolume{10}(\bissue{3}),
\bfpage{726}--\blpage{752}
(\byear{2000})
{\href{https://arxiv.org/abs/2667319}{{2667319}}}
\end{barticle}
\endbibitem

\bibitem[\protect\citeauthoryear{Chen
  et~al.}{1999}]{chenLiftingMarkovChains1999}
\begin{bchapter}
\bauthor{\bsnm{Chen}, \binits{F.}},
\bauthor{\bsnm{Lov{\'a}sz}, \binits{L.}},
\bauthor{\bsnm{Pak}, \binits{I.}}:
\bctitle{Lifting {{Markov}} chains to speed up mixing}.
In: \bbtitle{Proceedings of the Thirty-First Annual {{ACM}} Symposium on
  {{Theory}} of {{Computing}}},
pp. \bfpage{275}--\blpage{281}.
\bpublisher{ACM},
\blocation{Atlanta Georgia USA}
(\byear{1999}).
\doiurl{10.1145/301250.301315}
\end{bchapter}
\endbibitem

\bibitem[\protect\citeauthoryear{Bernard
  et~al.}{2009}]{bernardEventchainMonteCarlo2009}
\begin{barticle}
\bauthor{\bsnm{Bernard}, \binits{E.P.}},
\bauthor{\bsnm{Krauth}, \binits{W.}},
\bauthor{\bsnm{Wilson}, \binits{D.B.}}:
\batitle{Event-chain {{Monte Carlo}} algorithms for hard-sphere systems}.
\bjtitle{Physical Review E}
\bvolume{80}(\bissue{5}),
\bfpage{056704}
(\byear{2009})
\doiurl{10.1103/PhysRevE.80.056704}
\end{barticle}
\endbibitem

\bibitem[\protect\citeauthoryear{Michel
  et~al.}{2014}]{michelGeneralizedEventchainMonte2014}
\begin{barticle}
\bauthor{\bsnm{Michel}, \binits{M.}},
\bauthor{\bsnm{Kapfer}, \binits{S.C.}},
\bauthor{\bsnm{Krauth}, \binits{W.}}:
\batitle{Generalized event-chain {{Monte Carlo}}: {{Constructing}}
  rejection-free global-balance algorithms from infinitesimal steps}.
\bjtitle{The Journal of Chemical Physics}
\bvolume{140}(\bissue{5}),
\bfpage{054116}
(\byear{2014})
\doiurl{10.1063/1.4863991}
\end{barticle}
\endbibitem

\bibitem[\protect\citeauthoryear{Peters and
  de~With}{2012}]{peters2012rejection}
\begin{barticle}
\bauthor{\bsnm{Peters}, \binits{E.A.J.F.}},
\bauthor{\bsnm{With}, \binits{G.}}:
\batitle{Rejection-free monte carlo sampling for general potentials}.
\bjtitle{Physical Review E}
\bvolume{85}(\bissue{2}),
\bfpage{026703}
(\byear{2012})
\end{barticle}
\endbibitem

\bibitem[\protect\citeauthoryear{Gustafson}{1998}]{GUSTAFSON1998}
\begin{barticle}
\bauthor{\bsnm{Gustafson}, \binits{P.}}:
\batitle{A guided walk metropolis algorithm}.
\bjtitle{Statistics and Computing}
\bvolume{8}(\bissue{4}),
\bfpage{357}--\blpage{364}
(\byear{1998})
\doiurl{10.1023/A:1008880707168}
\end{barticle}
\endbibitem

\bibitem[\protect\citeauthoryear{Bierkens}{2016}]{bierkensNonreversibleMetropolisHastings2016}
\begin{barticle}
\bauthor{\bsnm{Bierkens}, \binits{J.}}:
\batitle{Non-reversible {{Metropolis-Hastings}}}.
\bjtitle{Statistics and Computing}
\bvolume{26}(\bissue{6}),
\bfpage{1213}--\blpage{1228}
(\byear{2016})
\doiurl{10.1007/s11222-015-9598-x}
\end{barticle}
\endbibitem

\bibitem[\protect\citeauthoryear{Kamatani and Song}{2023}]{Kamatani2023}
\begin{barticle}
\bauthor{\bsnm{Kamatani}, \binits{K.}},
\bauthor{\bsnm{Song}, \binits{X.}}:
\batitle{Non-reversible guided metropolis kernel}.
\bjtitle{Journal of Applied Probability}
\bvolume{60}(\bissue{3}),
\bfpage{955}--\blpage{981}
(\byear{2023})
\doiurl{10.1017/jpr.2022.109}
\end{barticle}
\endbibitem

\bibitem[\protect\citeauthoryear{Davis}{1984}]{davisPiecewiseDeterministicMarkovProcesses1984}
\begin{barticle}
\bauthor{\bsnm{Davis}, \binits{M.H.A.}}:
\batitle{Piecewise-{{Deterministic Markov Processes}}: {{A General Class}} of
  {{Non-Diffusion Stochastic Models}}}.
\bjtitle{Journal of the Royal Statistical Society: Series B (Methodological)}
\bvolume{46}(\bissue{3}),
\bfpage{353}--\blpage{376}
(\byear{1984})
\doiurl{10.1111/j.2517-6161.1984.tb01308.x}
\end{barticle}
\endbibitem

\bibitem[\protect\citeauthoryear{Davis}{1993}]{davisMarkovModelsOptimisation1993}
\begin{bbook}
\bauthor{\bsnm{Davis}, \binits{M.H.A.}}:
\bbtitle{Markov Models and Optimization}.
\bpublisher{Chapman and Hall},
\blocation{London}
(\byear{1993})
\end{bbook}
\endbibitem

\bibitem[\protect\citeauthoryear{{Bouchard-C{\^o}t{\'e}}
  et~al.}{2018}]{bouchard-coteBouncyParticleSampler2018}
\begin{barticle}
\bauthor{\bsnm{{Bouchard-C{\^o}t{\'e}}}, \binits{A.}},
\bauthor{\bsnm{Vollmer}, \binits{S.J.}},
\bauthor{\bsnm{Doucet}, \binits{A.}}:
\batitle{The {{Bouncy Particle Sampler}}: {{A Nonreversible Rejection-Free
  Markov Chain Monte Carlo Method}}}.
\bjtitle{Journal of the American Statistical Association}
\bvolume{113}(\bissue{522}),
\bfpage{855}--\blpage{867}
(\byear{2018})
\doiurl{10.1080/01621459.2017.1294075}
\end{barticle}
\endbibitem

\bibitem[\protect\citeauthoryear{Bierkens and
  Roberts}{2017}]{bierkensPiecewiseDeterministicScaling2017}
\begin{barticle}
\bauthor{\bsnm{Bierkens}, \binits{J.}},
\bauthor{\bsnm{Roberts}, \binits{G.}}:
\batitle{A piecewise deterministic scaling limit of lifted
  {{Metropolis}}--{{Hastings}} in the {{Curie}}--{{Weiss}} model}.
\bjtitle{The Annals of Applied Probability}
\bvolume{27}(\bissue{2}),
\bfpage{846}--\blpage{882}
(\byear{2017})
\doiurl{10.1214/16-AAP1217}
\end{barticle}
\endbibitem

\bibitem[\protect\citeauthoryear{Bierkens
  et~al.}{2019}]{bierkensZigZagProcessSuperefficient2019}
\begin{barticle}
\bauthor{\bsnm{Bierkens}, \binits{J.}},
\bauthor{\bsnm{Fearnhead}, \binits{P.}},
\bauthor{\bsnm{Roberts}, \binits{G.}}:
\batitle{The {{Zig-Zag}} process and super-efficient sampling for {{Bayesian}}
  analysis of big data}.
\bjtitle{The Annals of Statistics}
\bvolume{47}(\bissue{3}),
\bfpage{1288}--\blpage{1320}
(\byear{2019})
\doiurl{10.1214/18-AOS1715}
\end{barticle}
\endbibitem

\bibitem[\protect\citeauthoryear{Michel
  et~al.}{2020}]{michelForwardEventChainMonte2020}
\begin{barticle}
\bauthor{\bsnm{Michel}, \binits{M.}},
\bauthor{\bsnm{Durmus}, \binits{A.}},
\bauthor{\bsnm{S{\'e}n{\'e}cal}, \binits{S.}}:
\batitle{Forward {{Event-Chain Monte Carlo}}: {{Fast Sampling}} by {{Randomness
  Control}} in {{Irreversible Markov Chains}}}.
\bjtitle{Journal of Computational and Graphical Statistics}
\bvolume{29}(\bissue{4}),
\bfpage{689}--\blpage{702}
(\byear{2020})
\doiurl{10.1080/10618600.2020.1750417}
\end{barticle}
\endbibitem

\bibitem[\protect\citeauthoryear{Bierkens
  et~al.}{2020}]{bierkensBoomerangSampler2020}
\begin{bchapter}
\bauthor{\bsnm{Bierkens}, \binits{J.}},
\bauthor{\bsnm{Grazzi}, \binits{S.}},
\bauthor{\bsnm{Kamatani}, \binits{K.}},
\bauthor{\bsnm{Roberts}, \binits{G.}}:
\bctitle{The boomerang sampler}.
In: \beditor{\bsnm{Daumé}, \binits{H.} \bsuffix{III}},
\beditor{\bsnm{Singh}, \binits{A.}} (eds.)
\bbtitle{Proceedings of the 37th International Conference on Machine Learning}.
\bsertitle{Proceedings of Machine Learning Research},
vol. \bseriesno{119},
pp. \bfpage{908}--\blpage{918}.
\bpublisher{PMLR},
\blocation{Virtual}
(\byear{2020}).
\burl{https://proceedings.mlr.press/v119/bierkens20a.html}
\end{bchapter}
\endbibitem

\bibitem[\protect\citeauthoryear{Wu and
  Robert}{2020}]{wuCoordinateSamplerNonreversible2020}
\begin{barticle}
\bauthor{\bsnm{Wu}, \binits{C.}},
\bauthor{\bsnm{Robert}, \binits{C.P.}}:
\batitle{Coordinate sampler: A non-reversible {{Gibbs-like MCMC}} sampler}.
\bjtitle{Statistics and Computing}
\bvolume{30}(\bissue{3}),
\bfpage{721}--\blpage{730}
(\byear{2020})
\doiurl{10.1007/s11222-019-09913-w}
\end{barticle}
\endbibitem

\bibitem[\protect\citeauthoryear{Vasdekis and
  Roberts}{2023}]{vasdekisSpeedZigZag2022}
\begin{barticle}
\bauthor{\bsnm{Vasdekis}, \binits{G.}},
\bauthor{\bsnm{Roberts}, \binits{G.O.}}:
\batitle{{Speed up Zig-Zag}}.
\bjtitle{The Annals of Applied Probability}
\bvolume{33}(\bissue{6A}),
\bfpage{4693}--\blpage{4746}
(\byear{2023})
\doiurl{10.1214/23-AAP1930}
\end{barticle}
\endbibitem

\bibitem[\protect\citeauthoryear{Sutton and
  Fearnhead}{2023}]{suttonConcaveConvexPDMPbasedSampling2023}
\begin{barticle}
\bauthor{\bsnm{Sutton}, \binits{M.}},
\bauthor{\bsnm{Fearnhead}, \binits{P.}}:
\batitle{Concave-{{Convex PDMP-based Sampling}}}.
\bjtitle{Journal of Computational and Graphical Statistics}
\bvolume{32}(\bissue{4}),
\bfpage{1425}--\blpage{1435}
(\byear{2023})
\doiurl{10.1080/10618600.2023.2203735}
\end{barticle}
\endbibitem

\bibitem[\protect\citeauthoryear{Hardcastle
  et~al.}{2025}]{hardcastleAveragingPolyhazardModels2025}
\begin{barticle}
\bauthor{\bsnm{Hardcastle}, \binits{L.}},
\bauthor{\bsnm{Livingstone}, \binits{S.}},
\bauthor{\bsnm{Baio}, \binits{G.}}:
\batitle{Averaging polyhazard models using piecewise deterministic monte carlo
  with applications to data with long-term survivors}.
\bjtitle{The Annals of Applied Statistics}
\bvolume{19}(\bissue{4}),
\bfpage{3179}--\blpage{3202}
(\byear{2025})
\doiurl{10.1214/25-AOAS2064}
\end{barticle}
\endbibitem

\bibitem[\protect\citeauthoryear{Schauer and Grazzi}{2020}]{zigzagboomerangjl}
\begin{botherref}
\oauthor{\bsnm{Schauer}, \binits{M.}},
\oauthor{\bsnm{Grazzi}, \binits{S.}}:
{ZigZagBoomerang.jl}: v0.5.3.
Zenodo
(2020).
\doiurl{10.5281/zenodo.3931118} .
\url{https://doi.org/10.5281/zenodo.3931118}
\end{botherref}
\endbibitem

\bibitem[\protect\citeauthoryear{Pagani
  et~al.}{2024}]{paganiNuZZNumericalZigZag2024}
\begin{barticle}
\bauthor{\bsnm{Pagani}, \binits{F.}},
\bauthor{\bsnm{Chevallier}, \binits{A.}},
\bauthor{\bsnm{Power}, \binits{S.}},
\bauthor{\bsnm{House}, \binits{T.}},
\bauthor{\bsnm{Cotter}, \binits{S.}}:
\batitle{{{NuZZ}}: {{Numerical Zig-Zag}} for general models}.
\bjtitle{Statistics and Computing}
\bvolume{34}(\bissue{1}),
\bfpage{61}
(\byear{2024})
\doiurl{10.1007/s11222-023-10363-8}
\end{barticle}
\endbibitem

\bibitem[\protect\citeauthoryear{Bertazzi
  et~al.}{2022}]{bertazziApproximationsPiecewiseDeterministic2022}
\begin{barticle}
\bauthor{\bsnm{Bertazzi}, \binits{A.}},
\bauthor{\bsnm{Bierkens}, \binits{J.}},
\bauthor{\bsnm{Dobson}, \binits{P.}}:
\batitle{Approximations of {{Piecewise Deterministic Markov Processes}} and
  their convergence properties}.
\bjtitle{Stochastic Processes and their Applications}
\bvolume{154},
\bfpage{91}--\blpage{153}
(\byear{2022})
\doiurl{10.1016/j.spa.2022.09.004}
\end{barticle}
\endbibitem

\bibitem[\protect\citeauthoryear{Bradbury et~al.}{2018}]{jax2018github}
\begin{botherref}
\oauthor{\bsnm{Bradbury}, \binits{J.}},
\oauthor{\bsnm{Frostig}, \binits{R.}},
\oauthor{\bsnm{Hawkins}, \binits{P.}},
\oauthor{\bsnm{Johnson}, \binits{M.J.}},
\oauthor{\bsnm{Leary}, \binits{C.}},
\oauthor{\bsnm{Maclaurin}, \binits{D.}},
\oauthor{\bsnm{Necula}, \binits{G.}},
\oauthor{\bsnm{Paszke}, \binits{A.}},
\oauthor{\bsnm{Vander{P}las}, \binits{J.}},
\oauthor{\bsnm{Wanderman-{M}ilne}, \binits{S.}},
\oauthor{\bsnm{Zhang}, \binits{Q.}}:
{JAX}: composable transformations of {P}ython+{N}um{P}y programs
(2018).
\url{https://github.com/jax-ml/jax}
\end{botherref}
\endbibitem

\bibitem[\protect\citeauthoryear{Vanetti
  et~al.}{2018}]{vanettiPiecewiseDeterministicMarkovChain2018}
\begin{botherref}
\oauthor{\bsnm{Vanetti}, \binits{P.}},
\oauthor{\bsnm{Bouchard-Côté}, \binits{A.}},
\oauthor{\bsnm{Deligiannidis}, \binits{G.}},
\oauthor{\bsnm{Doucet}, \binits{A.}}:
Piecewise-Deterministic Markov Chain Monte Carlo
(2018).
\url{https://arxiv.org/abs/1707.05296}
\end{botherref}
\endbibitem

\bibitem[\protect\citeauthoryear{Wu and Robert}{2017}]{Wu2017}
\begin{botherref}
\oauthor{\bsnm{Wu}, \binits{C.}},
\oauthor{\bsnm{Robert}, \binits{C.P.}}:
Generalized Bouncy Particle Sampler
(2017).
\url{https://arxiv.org/abs/1706.04781}
\end{botherref}
\endbibitem

\bibitem[\protect\citeauthoryear{Lewis and
  Shedler}{1979}]{lewisSimulationNonhomogeneousPoisson1979}
\begin{barticle}
\bauthor{\bsnm{Lewis}, \binits{P.a.W.}},
\bauthor{\bsnm{Shedler}, \binits{G.S.}}:
\batitle{Simulation of nonhomogeneous poisson processes by thinning}.
\bjtitle{Naval Research Logistics Quarterly}
\bvolume{26}(\bissue{3}),
\bfpage{403}--\blpage{413}
(\byear{1979})
\doiurl{10.1002/nav.3800260304}
\end{barticle}
\endbibitem

\bibitem[\protect\citeauthoryear{Bierkens et~al.}{2019}]{10.1214/18-AAP1453}
\begin{barticle}
\bauthor{\bsnm{Bierkens}, \binits{J.}},
\bauthor{\bsnm{Roberts}, \binits{G.O.}},
\bauthor{\bsnm{Zitt}, \binits{P.-A.}}:
\batitle{{Ergodicity of the zigzag process}}.
\bjtitle{The Annals of Applied Probability}
\bvolume{29}(\bissue{4}),
\bfpage{2266}--\blpage{2301}
(\byear{2019})
\doiurl{10.1214/18-AAP1453}
\end{barticle}
\endbibitem

\end{thebibliography}

\end{document}